\documentclass[runningheads]{llncs}
\usepackage[T1]{fontenc}
\usepackage{graphicx}
\usepackage{soul}
\usepackage{url}
\usepackage[utf8]{inputenc}
\usepackage[small]{caption}
\usepackage{graphicx}
\usepackage{xcolor}
\usepackage{amsmath}
\usepackage{booktabs}
\usepackage[linesnumbered,ruled,vlined]{algorithm2e}
\usepackage{enumerate}
\usepackage{amssymb}
\usepackage{tikz}
\usepackage{bbm}

\usepackage{todonotes}
\presetkeys{todonotes}{inline}{}

\DeclareMathOperator*{\argmin}{arg\,min}
\DeclareMathOperator*{\argmax}{arg\,max}

\begin{document}
\title{Fair Division of Chores with Budget Constraints}
%
%
\author{Edith Elkind\inst{1,2} \and
Ayumi Igarashi\inst{3} \and
Nicholas Teh\inst{1}}
\authorrunning{E. Elkind et al.}
%
\institute{University of Oxford, UK \and
Alan Turing Institute, UK \and
University of Tokyo, Japan}
\maketitle              
\begin{abstract}
We study fair allocation of indivisible chores to agents under budget constraints, where each chore has an objective size and disutility. This model captures scenarios where a set of chores need to be divided among agents with limited time, and each chore has a specific time needed for completion. We propose a budget-constrained model for allocating  indivisible chores, and systematically explore the differences between goods and chores in this setting.
We establish the existence of an EFX allocation. We then show that EF2 allocations are polynomial-time computable in general; for many restricted settings, we strengthen this result to EF1.
For divisible chores, we develop an efficient algorithm for computing an EF allocation.

\end{abstract}
\section{Introduction}
Alice, Bob, and their teenage children Claire and Dan wish to fairly divide a set of household chores. Each chore requires a certain amount of time to complete; for simplicity, assume that this amount, as well as the disutility of the chore,  does not depend on who performs the chore (this is approximately true for many chores). Alice works long shifts, so she only has 5 hours a week to dedicate to chores. Bob has a more conventional schedule, so he can spend 10 hours on chores. Claire and Dan have many extracurricular activities, so they can contribute 7 and 4 hours, respectively. 
How can they divide the chores in a way that is fair and respects time constraints?

In a similar vein, consider a company that needs to allocate several time-consuming tasks to a group of employees, in addition to their regular workload.
As employees may have different existing workloads, the amounts of extra work they would be able to take on differ as well.

In both of our examples, it is not immediately clear what it means to be fair, given agents' different time budgets. Thus, we need to adapt the notions of fairness that have been developed in the fair division literature to our setting, and then determine under what conditions fair allocations exist and whether they can be computed in polynomial time. While the budgeted setting has been considered for goods (see Section~\ref{sec:related} for a discussion of related work), extending ideas from prior work to chores poses new challenges.

\subsection{Our Contributions}
We introduce a framework for allocating chores with objective (i.e., agent-inde\-pen\-dent) sizes and disutilities under budget constraints. 
In Section 2, we set up our formal model, put forward notions of fairness that are appropriate for this setting, and discuss the challenges that arise when adapting the budget-constrained model from goods to chores.

In Section 3, we show the existence of an EFX allocation for indivisible chores. 
Perhaps surprisingly, an adaptation of the EFX algorithm for goods with objective sizes and utilities in the budgeted setting also works for our scenario; we note that this is not universally true for general restricted instances where EFX is known to exist for goods. 
This is particularly interesting because EFX for chores is incomparable to EFX for goods, in the sense that the special cases where EFX allocations are known to exist are quite different in these two settings. Moreover, 
techniques for proving EFX in the goods setting are also known to be very different beyond two agents or identical valuations (unlike for EF1).

In Section 4, we provide a polynomial-time algorithm for computing EF2 allocations. 
Section 5 then looks at five special cases---when chores are identically-valued, identically-sized, identically-dense, agents have identical budgets, and the case of two agents---for which we can compute EF1 allocations efficiently. 
Most of the results in the above two sections rely on a greedy ``densest-item-first'' algorithm. 
In the goods case, the greedy algorithm that achieves similar guarantees is (surprisingly) also a densest-item first greedy algorithm, even though, intuitively, high density is desirable for  goods and undesirable for chores. This suggests that the symmetry between chores and goods sometimes presents itself in unexpected ways. 

Section 6 focuses on divisible chores, and establishes the existence of EF allocations.




\subsection{Related Work}\label{sec:related}
The mathematical framework for fair division has been put forward by Steinhaus~\cite{Steinhaus48} over 70 years ago, and this field has seen an explosion
of interest in recent years (see, e.g., a survey by Amanatidis et al.~\cite{amanatidis2023survey}). Historically, 
most works in the field focused on the allocation of \emph{goods}, i.e., items that  are valued non-negatively by all agents.
While some of the results for fair allocation of goods extend easily to chores, there are many real-world applications for which this is not the case. This observation led to a recent line of work that considers the allocation of \emph{chores}, i.e., items which agents value negatively; see, e.g.,~\cite{AzizCaIg22chores,bogomolnaia2017competitive,dehghani2018envy} as well as the recent survey by Aziz et al.~\cite{aziz2022survey}.
Indeed, allocating chores is generally known to be more difficult than allocating goods, with more open problems in chores than goods, and many techniques that work for goods, but do not directly translate to chores; a notable example here is maximization of Nash social welfare.
Also, a number of authors have considered the problem of allocating goods to agents under budget constraints~\cite{BarmanKhShSrAAAI2023,GanLiWu2023,GarbeaGkTa2023,GanLiWuIJCAI2021}, as well as rent division with budget constraints~\cite{airiau2023rentbudget,procaccia2018rentbudget}.
However, 
to the best of our knowledge, we are the first to explore the intersection of these two lines of work, i.e., chore allocation under budget constraints.

Additional motivation for our analysis is provided by the recent work of Igarashi and Yokoyama~\cite{igarashi2023choreapp}, who have developed an application that helps couples to fairly divide household chores. While their tool captures many aspects of the task, it does not allow the household members to specify budget constraints, and this reduces the usability of the tool. We believe that incorporating such constraints will help many households to come up with a better way of sharing the workload.

We will now discuss prior work on the allocation of goods under budget constraints in more detail.
Gan et al.~\cite{GanLiWu2023,GanLiWuIJCAI2021} assumed identical valuations and size functions, and studied approximation ratios (with respect to EF1) for the maximum Nash welfare rule in budget-constrained scenarios. In particular, they proposed an approximation algorithm for achieving 1/2-EF1, along with special cases where they could guarantee EF1.
Barman et al.~\cite{BarmanKhShSrAAAI2023} considered the same model and proposed an algorithm that satisfies EF2 in general, and EF1 in special cases.

Garbea et al.~\cite{GarbeaGkTa2023} were the first to consider the budget-constrained model with subjective valuation functions (but still identical size functions).
However, their results are limited to two- and three-agent cases. 
They design algorithms that guarantee EFX, while achieving approximation of Nash welfare for these special cases.

Barman et al.~\cite{BarmanKhShSr2023} studied a more general model of fairly allocating goods under \emph{generalized assignment constraints}, extending the traditional budget-con\-strained model to one where the sizes and values of the goods can be subjective.
They showed the existence (via a pseudopolynomial time algorithm) of EFX allocations for indivisible goods case.

\section{Preliminaries}
\label{sec:preliminaries}
For each positive integer $z$, let $[z] := \{1,\dots,z\}$.
Let $N = [n]$ be a set of $n$ \emph{agents} and $C = \{c_1,\dots,c_m\}$ be a set of $m$ 
\emph{chores}.
 Each chore $c \in C$ has an objective \emph{size} $s(c) \in \mathbb{R}_{> 0}$ and \emph{disutility} $d(c) \in \mathbb{R}_{\geq 0}$; we write $\rho(c) = \frac{d(c)}{s(c)}$ to denote the \emph{density} of the chore $c$.
Each agent has a \emph{budget} $B_i \in \mathbb{R}_{> 0}$; let $\mathbf{B} = (B_1,\dots,B_n)$ be the vector of agents' budgets. For our algorithmic results, we assume that all sizes, disutilities and budgets are rational numbers given in binary.

Unlike in the unconstrained fair allocation model, in our setting it may be impossible to divide all the chores among the agents in $N$: e.g., the sum of sizes may exceed the sum of the agents' budgets. As we cannot simply discard the chores, this necessitates the introduction of a {\em housekeeper}, whose role is similar to that of charity in the budget-constrained model for allocating goods.
It is assumed that the housekeeper is paid to complete the chores; this payment is exogenous to the model, and an external consideration.
Our fairness notions are formulated in such a way that as few chores as possible are allocated to the housekeeper; this is similar to how the allocation to charity is treated in a goods context. 

We set up the model for indivisible and divisible chores separately.

\subsection{Indivisible Chores}

A \emph{bundle} of chores is a subset of $C$.
We assume that sizes and disutilities of chores are additive, so that for each bundle $S \subseteq C$ its size $s(S)$ and disutility $d(S)$ are given by, respectively, $s(S) = \sum_{c \in S} s(c)$ and $d(S) = \sum_{c \in S} d(c)$. 
This assumption is standard across all works dealing with budget constraints, and is also common for fair division problems in general. 

An \emph{allocation} $\mathcal{A} = (A_1,\dots,A_{n+1})$ is a partition of $C$ into $n+1$ disjoint bundles of chores, where $A_i$ is assigned to agent $i$, and $A_{n+1}$ is the set of unallocated chores; we will refer to $A_{n+1}$ as the bundle allocated to the housekeeper. 
We say that an allocation $\mathcal{A}$ is \emph{feasible} if $s(A_i) \leq B_i$ for all $i \in [n]$.

Next, we define several notions of fairness for our setting.
Our definitions mirror the respective definitions for allocating goods under budget constraints.
\begin{definition}[Envy-freeness] \label{def:EF:chores}
    An allocation $\mathcal{A} = (A_1,\dots,A_n, \allowbreak A_{n+1})$ is said to be \emph{envy-free (EF)} if for all $i\in [n+1]$, $j \in [n]$ and for every subset $S \subseteq A_i$ with $s(S) \leq B_j$ it holds that $d(S) \leq d(A_j)$.
\end{definition}
Intuitively, an allocation is envy-free if for every agent $i\in [n]$ as well as for the housekeeper it holds that if they consider a subset $S$ of their bundle that could be allocated to an agent $j\in [n]$ (in the sense of having a size that does not exceed $B_j$), they find $S$ to be at most as unpleasant/objectionable as the actual bundle of~$j$. 

Even in the absence of budget constraints, it may be impossible to allocate indivisible items (chores or goods) in an envy-free manner.
Clearly, this negative result also applies to the setting with budget constraints.
This observation motivates us to adapt several relaxations of EF  to our setting. 
We first consider a popular, relatively strong relaxation of EF, which has been widely studied in the unconstrained setting.
\begin{definition}[Envy-freeness up to any chore] \label{def:EFX:chores}
    An allocation $\mathcal{A} = (A_1,\dots,\allowbreak A_n,A_{n+1})$ is said to be \emph{envy-free up to any chore (EFX)} if for all $i\in [n+1]$, $j \in [n]$, for every subset $S \subseteq A_i$ with $s(S) \leq B_j$, and for each $c \in S$ it holds that $d(S \setminus \{c\}) \leq d(A_j)$.
\end{definition}
Next, we consider another class of relaxations of EF.
\begin{definition}[Envy-freeness up to $k$ chores] \label{def:EF1:chores}
    Given a positive integer $k$, an allocation $\mathcal{A} = (A_1,\dots,\allowbreak A_n,A_{n+1})$ is said to be \emph{envy-free up to $k$ chores (EF$k$)} if for every $i\in [n+1]$, $j \in [n]$, and for every subset $S \subseteq A_i$ with $s(S) \leq B_j$ there exists a subset $S' \subseteq S$ with $|S'| = k$ such that $d(S \setminus S') \leq d(A_j)$.
\end{definition}
The most commonly studied property is EF1 (i.e., EF$k$ with $k=1$).
However, following the analysis of Barman et al.~\cite{BarmanKhShSr2023} in the budget-constrained goods setting, 
we will also consider EF2 (i.e., EF$k$ with $k=2$).


Note that in our definitions of (approximate) envy-freeness $i$ takes values in $[n+1]$ rather than $[n]$, i.e., we want the housekeeper to be (approximately) non-envious towards the agents. This ensures that \emph{sufficiently} many chores are allocated to agents: e.g., an allocation where all chores are allocated to the housekeeper is not envy-free unless no agent can execute any of the chores.

\subsection{Divisible Chores}
Next, we consider the model for divisible chores.
To the best of our knowledge, budget-constrained allocation of divisible goods was only considered by Barman et al.~\cite{BarmanKhShSr2023}, albeit under the more general \emph{generalized assignment constraints} (where size and disutility functions can be subjective).

A \emph{fractional bundle} of chores is an $m$-dimensional vector $X = (x_{1}, x_{2},\dots, x_{m})$ with $x_j\in [0, 1]$ for each $j\in [m]$; here, $x_j$ is the fraction of chore $c_j\in C$ that is placed in $X$.
The sizes and disutilities of fractional chores are assumed to be additive, i.e., we write $s(X) = \sum_{j \in [m]} x_{j} \cdot s(c_j)$ and $d(X) = \sum_{j \in [m]} x_{j} \cdot d(c_j)$. Again, this assumption is standard in the literature.

An {\em allocation} $\mathcal{X} = (X_1,\dots,X_{n+1})$ is a list of $n+1$ fractional bundles of chores 
with $X_i=(x_{i1}, \dots, x_{im})$ for each $i\in [n+1]$
that satisfies $\sum_{i \in [n+1]} x_{ij} = 1$. We refer to
$X_i$ as the fractional bundle of agent $i$; $X_{n+1}$ is the fractional bundle allocated to the housekeeper.
We say that an allocation is \emph{feasible} if for all $i \in [n]$ it holds that $s(X_i) \leq B_i$.

To formulate a notion of envy-freeness in this setting, we first define a vector comparison operator that enables us to compare fractional bundles of chores. 
Specifically, for any pair of vectors $W = (w_1,\dots,w_m), Y = (y_1,\dots,y_m) \in [0,1]^m$, we write $W \leq Y$ if and only if $w_j \leq y_j$ for all $j \in [m]$.

Now, we define envy-freeness for divisible chores under budget constraints.
Again, this definition mirrors the one for the budget-constrained goods setting.
\begin{definition}[Envy-freeness for divisible chores] \label{def:efdiv}
    An allocation $\mathcal{X} = (X_1, \allowbreak\dots, X_{n+1}) \in [0,1]^{(n+1)\times m}$ of divisible chores is said to be \emph{envy-free (EF)} if for every $i \in [n+1], j\in [n]$ and for all fractional assignments $Y \leq X_i$ with $s(Y) \leq B_j$ it holds that $d(Y) \leq d(X_j)$.
\end{definition}
We will show that, unlike in the indivisible setting, 
divisible chores always admit 
an envy-free allocation (this is also the case for divisible goods).
Therefore, defining further relaxations is not necessary.

\subsubsection*{Modeling Assumptions: A Discussion}
Our formal model considers only objective (i.e., identical) size and disutility functions.
Of course, in practice different agents may assign different disutilities to the same chore: while Alice dislikes dusting more than doing dishes, 
Bob has the opposite preferences. It may also be the case that the size of the chore varies from one agent to another: while Alice can peel potatoes for dinner in 5 minutes, Bob will spend 8 minutes on the same task. However, we chose to leave modeling non--identical disutilities and sizes in the budgeted setting to future work. The reasons for this decision are as follows.

First, as noted by Barman et al.~\cite{BarmanKhShSrAAAI2023}, considering budget constraints even under identical valuation functions already constitutes a technically-rich model, due to the additional size (and budget) dimension.

Second, the more general formulation, where agents have subjective size functions (even under identical valuation functions), does not admit a polynomial-time approximation scheme for the value-maximization objective \cite{chekuri2005ptas}.

Third, we have argued that it is necessary to introduce the housekeeper agent, and there is no principled way to define the disutility function for the housekeeper if the agents' disutility functions are non-identical. However, to define (relaxations of) envy-freeness, we would have to reason about the housekeeper's disutility.

We note that most of the existing works studying the allocation of indivisible goods under budget constraints \cite{BarmanKhShSrAAAI2023,GanLiWu2023,GanLiWuIJCAI2021} assume that each good has an objective size, and that agents have an objective valuation function.
Models with identical valuations have also been widely studied in the setting of goods without budget constraints \cite{barman2020identical,mutzari2023resilient,plaut2018efx}.
An important exception is a recent paper by Barman et al.~\cite{BarmanKhShSr2023}, who extend the concepts for the allocation of goods to generalized assignment constraints (as opposed to budget constraints), where sizes are allowed to be agent-specific. They showed the existence (but not polynomial-time computability) of EFX allocations for indivisible goods. However, extending their definitions and results to the setting of chores with non-identical sizes is not straightforward.
 
\section{Existence of EFX Allocations for Indivisible Chores}
To begin, we consider the existence of EFX allocations for indivisible chores under budget constraints.

The existence of EFX allocations for more than three agents in the indivisible goods allocation setting is a longstanding open problem in fair division.
The setting of chores has been shown to be even more difficult (refer to the surveys of Aziz et al.~\cite{aziz2022survey} and Amanatidis et al.~\cite{amanatidis2023survey}).
Recently, Barman et al.~\cite{BarmanKhShSr2023} proved the existence of EFX allocations for indivisible goods under budget constraints, for the case of identical disutility functions.
Their algorithm is a close adaptation of the algorithm for finding EFX allocations in restricted settings \cite{chaudhury2021little}.
In this section, we extend this positive result to the case of chores.


We first introduce the concept of a \emph{manageable set}, which is similar in spirit to the concept of a \emph{minimal envied subset}; the latter is used to prove the existence of EFX allocations for indivisible goods in various restricted settings  \cite{BarmanKhShSr2023,chaudhury2021little,ghosal2023efxfour}.

\begin{definition}
    [Manageable set] \label{def:manageable}
    A set of chores $T \subseteq C$ is said to be a \emph{manageable set} for an allocation $\mathcal{A} = (A_1,\dots, \allowbreak A_n, A_{n+1})$ if 
    \begin{enumerate}[(i)]
        \item there exists an $i\in [n]$ such that $s(T)\le B_i$  and $d(T) > d(A_i)$, and
        \item 
        no strict subset of $T$ satisfies~(i), i.e., 
        for each strict subset $T' \subsetneq T$ and each $k \in [n]$, either $s(T') > B_{k}$ or $s(T') \leq B_{k}$ and $d(T') \leq d(A_{k})$.
    \end{enumerate}
\end{definition}

Then, consider the following algorithm, which repeatedly finds a manageable set within the housekeeper's bundle and allocates it to one of the agents in $N$ in a feasible way.
 The algorithm terminates when the housekeeper's bundle no longer contains a manageable set.

\begin{algorithm}
\caption{Computes an EFX allocation}
\label{alg:efx}
\SetInd{0.8em}{0.3em}
{\bf Input} disutility function $d$, size function $s$, budgets $\mathbf{B}$\;
Initialize the allocation $\mathcal{A} =(A_1,\ldots,A_n,A_{n+1})=(\emptyset,\ldots,\emptyset,C)$\;
\While{there exists a subset $S \subseteq A_{n+1}$ such that $s(S) \leq B_i$ and $d(S) > d(A_i)$ for some $i \in [n]$\label{line:alg:efx}} 
{
Select a manageable set $T \subseteq A_{n+1}$ and a $k \in [n]$ with $s(T)\le B_k$, $d(T) > d(A_k)$\;
Update bundles $A_k \leftarrow T$ and $A_{n+1} \leftarrow C \setminus (\cup_{i\in [n]} A_i)$\;
} 
\Return{allocation $\mathcal{A}$}\;
\end{algorithm}

We will now show that Algorithm~\ref{alg:efx} returns an EFX allocation, thereby establishing that an EFX allocation is guaranteed to exist.
\begin{theorem}\label{thm:efx}
    Algorithm \ref{alg:efx} returns an \emph{EFX} allocation.
\end{theorem}
\begin{proof}
    First, we note that if the condition of the {\bf while} loop is satisfied, then $A_{n+1}$ contains a manageable set. Indeed, consider a minimum-size set $S$ that satisfies the condition in the {\bf while} loop for some $i\in [n]$.
    We have $s(S)\le B_i$, $d(S)> d(A_i)$, so $S$ satisfies condition (i) in the definition of a manageable set. Moreover, by our choice of $S$ no proper subset of $S$ satisfies (i), which means that $S$ satisfies condition (ii) as well.
    Thus, Algorithm~\ref{alg:efx} can indeed select a manageable set in line~4.
 
    Next, we observe that Algorithm \ref{alg:efx} necessarily terminates. Indeed, if at iteration $t$ we change the bundle of an agent $k\in [n]$ to $T$ in line~5, then $k$'s disutility increases (since before that step we had $d(T) > d(A_k)$) while the disutility of other agents in $[n]$ remains the same. Thus, the sum of disutilities of agents in $[n]$ goes up with each iteration.

    Further, once Algorithm \ref{alg:efx} terminates, the condition of the {\bf while} loop is no longer satisfied, which means that the housekeeper is not envious towards agents in $[n]$. It remains to argue that the EFX condition is satisfied for all other agents.

    Let $\mathcal{A}^{(t)} = (A_1^{(t)},\dots, A_{n+1}^{(t)})$ denote the allocation maintained by the algorithm just \emph{before} the $t$-th iteration  of the \textbf{while} loop (Line \ref{line:alg:efx}).
    We have $\mathcal{A}^{(1)} = (\emptyset,\dots,\emptyset, C)$ and $A_{n+1}^{(t)} = C \setminus (\cup_{i=1}^n A_i^{(t)})$ for all $t>0$.
    We will write $\mathcal{A}^{(t)}_n$ to denote the vector formed by the first $n$ entries of the vector $\mathcal{A}^{(t)}$ (i.e., excluding the bundle of the housekeeper).
    
    To show that the allocation $\mathcal{A}_n$ is EFX, we use induction.
    For the base case, note that in the first iteration, EFX trivially holds, as $\mathcal{A}^{(1)}_n = (\emptyset, \dots, \emptyset)$.
    This allocation is also feasible.
    Now, consider an iteration $t > 1$.
    For the inductive step, assume that the allocation $\mathcal{A}^{(t)}_n$ is feasible and satisfies EFX.
    In the $t$-th iteration, the algorithm changes the bundle of exactly one agent $k \in [n]$, by replacing it with $T$. The bundle $T$ satisfies $s(T)\le B_k$, so this update results in a feasible allocation.
    
    Since the bundles of all agents in $[n] \setminus \{k\}$ remain unchanged, the EFX condition is satisfied for each pair of agents not involving $k$.
    Thus, it remains to consider envy by/towards agent $k$ after $k$ has been allocated the bundle $T$.

    We first consider the envy experienced by agent $k$. Fix a subset $S\subseteq T$ and an agent $k'\in [n]\setminus\{k\}$ such that $s(S)\le B_{k'}$. For each $c\in S$ the set $S\setminus\{c\}$ is a proper subset of $T$.
    Since $T$ is a manageable set and $s(S\setminus\{c\})\le s(S)\le B_{k'}$, condition (ii) of Definition~\ref{def:manageable} implies that $d(S\setminus\{c\})\le d(A_{k'})$, i.e., 
    the envy by agent $k$ towards $k'$ can be eliminated by the removal of any single chore.

    Now, consider the envy of agent $k'\in [n]\setminus\{k\}$ towards agent $k$.
    By the induction hypothesis, before agent $k$ was allocated the bundle $T$, 
    the envy by $k'$ towards $k$ could be eliminated by removing a single chore.
    Moreover, $T$ has a higher disutility than the previous bundle of $k$. 
    Thus, for every subset $S \subseteq A_{k'}^{(t)} = A_{k'}^{(t+1)}$ with $s(S)\le B_k$ and every  $c\in S$ we have
    \begin{equation*}
        d(S\setminus\{c\}) \leq d(A_k^{(t)}) < d(A_k^{(t+1)}).
    \end{equation*}
    This implies that the envy by an agent $k' \in [n] \setminus \{k\}$ towards agent $k$ can be eliminated by the removal of any single chore.
    $\hfill \square$
\end{proof}
Theorem~\ref{thm:efx} establishes the existence of an EFX allocation for indivisible chores under budget constraints. While this result is constructive, the running time of Algorithm~\ref{alg:efx} is pseudopolynomial (to see this, note that checking the condition in Line~3 of the algorithm and finding a manageable set reduces to solving polynomially many instances of Knapsack) rather than polynomial.
The existence of a polynomial-time algorithm for computing an EFX allocation for chores remains an open problem (as with the setting of goods).


\section{Computing EF2 Allocations for Indivisible Chores}
Given that we do not know how to find EFX allocations in polynomial time, 
a natural follow-up direction is then to look for allocations that satisfy weaker relaxations of EF, but 
can be computed by algorithms that run in polynomial time. 
The most popular relaxation of EF after EFX would be EF1. 
However, even in the setting of allocating goods under budget constraints, the existence of a polynomial-time algorithm for computing EF1 allocations is still a challenging open question \cite{BarmanKhShSr2023,GanLiWu2023,GanLiWuIJCAI2021}.
Consequently, we shift our focus to a property that can be accomplished in polynomial time in the goods case, namely, EF2. Our next result shows that we can replicate this result for chores.

Specifically, it turns out that EF2 allocations can be found by the \textsc{DensestFirst} algorithm (Algorithm~\ref{alg:ef1_densestfirst}), which at each iteration picks an agent with minimum disutility and allocates to her the maximally-dense chore.
Interestingly, in the goods case a `densest-item first' greedy algorithm
also produces an EF2 allocation~\cite{barman2020identical}. This is 
surprising, because in the goods setting high-density items are particularly attractive, while in the chores setting high-density items are unattractive. Thus, while one may expect goods and chores to be symmetric, 
identifying the `correct' mapping from goods to chores is a non-trivial task.

\begin{algorithm}
\caption{\textsc{DensestFirst}}
\label{alg:ef1_densestfirst}
\SetInd{0.8em}{0.3em}
{\bf Input} disutility function $d$, size function $s$, budgets $\mathbf{B}$\;
Initialize the allocation $\mathcal{A} =(A_1,\ldots,A_n,A_{n+1})=(\emptyset,\ldots,\emptyset,C)$ and the set of live agents $L = [n]$\;
\While{$L \neq \emptyset$ and $A_{n+1} \neq \emptyset$} 
{
Let $i :=\min \argmin_{i \in L} d(A_i)$\; \label{alg:line:densestfirst:agent}
\If{for all $c \in A_{n+1}$ it holds that $s(A_{i^*} \cup \{c\}) > B_{i^*}$}{Remove $i^*$ from the set of live agents, i.e., $L \leftarrow L \setminus \{i^*\}$ \;}
\Else{
Choose a maximally-dense chore $c^* \in \argmax_{c \in C: s(A_{i^*}\cup \{c\}) \leq B_{i^*}} \rho(c)$, breaking ties in favor of smaller chores\; \label{alg:line:tie}
Update bundles $A_{i^*} \leftarrow A_{i^*} \cup \{c^*\}$ and $A_{n+1} \leftarrow A_{n+1} \setminus \{c^*\}$\;
}
} 
\Return{allocation $\mathcal{A}$}\;
\end{algorithm}

We first prove that the algorithm runs in polynomial-time.
\begin{theorem} \label{thm:densestfirst_runtime}
    Algorithm \ref{alg:ef1_densestfirst} runs in time $\mathcal{O}((n+m)^2)$.
\end{theorem}
\begin{proof}
    At each iteration of the \textbf{while} loop (Line 3), either an agent is removed from the set of live agents $L \subseteq N$, or one chore is removed from $A_{n+1}$.
    Thus, there can be at most $m+n$ iterations.
    The operation of finding the set of  agents in $L$ with minimum disutility (Line~\ref{alg:line:densestfirst:agent}) takes $\mathcal{O}(n)$ time. 
    Deciding if $i^*$ can be allocated an additional chore (Line 5) 
    and finding a maximally dense chore that is feasible for $i^*$ (Line 8) takes $\mathcal{O}(m)$ time.
    Together, the algorithm runs in $\mathcal{O}((n+m)^2)$ time.
    $\hfill \square$
\end{proof}

Next, we proceed to the main result of this section.
Due to space constraints, the proof is deferred to the full version of the paper.
\begin{theorem} \label{thm:ef2}
    Algorithm \ref{alg:ef1_densestfirst} returns an \emph{EF2} allocation.
\end{theorem}

The above result, coupled with the fact that the complexity of finding EF1 allocations in the budget-constrained goods model is still an open question, leads us to the next natural question: under what circumstances (i.e., special cases) can we compute EF1 allocations in polynomial time?
We investigate this in the next section.


\section{Computing EF1 Allocations in Special Cases for Indivisible Chores}
We consider five special cases where an EF1 allocation can be computed in polynomial time.
More specifically, we show that when chores are identically-valued (equivalently, when agents have binary disutility functions), identically-sized, identically-dense, or when agents have identical budgets, the \textsc{DensestFirst} algorithm (Algorithm \ref{alg:ef1_densestfirst}), which guarantees an EF2 allocation in our general model, will, in fact, return an EF1 allocation.
For each of these variants, we will only prove correctness, as the polynomial running time has already been established in Theorem \ref{thm:densestfirst_runtime}.
We then propose a separate polynomial-time algorithm that returns an EF1 allocation for two agents.

\subsection{Binary Disutility Functions or Identically-Valued Chores}
The first special case that we will look at is when agents have binary disutility functions (i.e., each chore is valued at either $0$ or $1$ by all agents).
Together with the identical valuation assumption, the case of binary disutilities reduces to that of \emph{identically-valued} chores.
This is because we can assume chores valued at $0$ is left unallocated (i.e., left in the housekeeper's bundle).
In this setting, it suffices to assume that, without loss of generality, each chore has a disutility of $1$.

While identical chores are trivial in the traditional fair division model (by simply allocating any $\lfloor m/n \rfloor$ chores to each agent, and
then picking an arbitrary set of $m-n\cdot \lfloor m/n \rfloor$ of agents and allocating each of these agents one of the remaining chores), the size dimension of the budget-constrained model leads to EF1 becoming a non-trivial property to prove.

We will now show that  executing Algorithm \ref{alg:ef1_densestfirst} on an instance with identically-valued chores results in an EF1 (in this case, equivalently, EFX) allocation.

\begin{theorem}\label{thm:binary}
    When agents have binary disutilities or when chores are identically-valued, Algorithm \ref{alg:ef1_densestfirst} returns an \emph{EF1} allocation.
\end{theorem}
\begin{proof}
    Consider any two agents $i,j \in [n]$ with $B_i \leq B_j$ and let $c_i^t$ and $c_j^t$ denote the $t$-th chore added to agent $i$ and $j$'s bundles, respectively.
    Also, let $A_i^t$ and $A_j^t$ denote the bundles belonging to agents $i$ and $j$, respectively after the $t$-th chore was added to their bundles.
    
    We first consider the envy of agent $i$ towards agent $j$. 
    By the fact that $B_i \leq B_j$ and by construction of the algorithm, we have $|A_i| \leq |A_j|+1$. Thus, $d(A_i) \leq  d(A_j)$ or $d(A_i \setminus \{c\}) \leq d(A_j)$ for some $c \in A_i$, which establishes the EF1 property by $i$ towards $j$.

    Next, we consider the envy by agent $j$ towards agent $i$.
    If $|A_j| \leq |A_i|+1$, then $d(A_j \setminus \{c\}) = |A_j|-1 \leq |A_i| = d(A_i)$ for any $c \in A_j$ and EF1 is trivially obtained.
    Thus, we assume $|A_j| > |A_i|+1$.

    Let $c^\alpha_i$ be the last chore that agent $i$ received. Consider the following two cases. 
    Since $|A_j| > |A_i|+1$, we have $|A_j| > |A_i| =  \alpha + 1$.
    \begin{description}
        \item[Case 1: $i < j$.]  
        Since $|A_j| > \alpha$, we have that $s(A_i^\alpha \cup \{c_j^{\alpha + 1}\}) > B_i$ (otherwise $c_j^{\alpha + 1}$ would have been allocated to agent $i$ instead).
        Then, we get that
        \begin{equation*}
            s(A_j^{\alpha+1}) = s(A_j^\alpha \cup \{c_j^{\alpha+1}\}) \geq s(A_i^\alpha \cup \{ c_j^{\alpha+1} \})  > B_i.
        \end{equation*}
        Consider any subset $S \subseteq A_j$ with $s(S) \leq B_i$. If $|S| > |A_j^\alpha|$, then $S$ contains at least $\alpha+1$ chores; however, since $A_j^{\alpha+1}$ contains the $\alpha+1$ smallest chores in $A_j$, this means that  $B_i<s(A_j^{\alpha+1}) \leq s(S)$, a contradiction. Thus, we have $|S| \leq |A_j^\alpha|$, giving us 
        \begin{equation*}
            d(S) = |S| \leq |A_j^\alpha| = |A_i^\alpha| = |A_i| = d(A_i).
        \end{equation*}
        
        \item[Case 2: $j < i$.]
        Since $|A_j| > \alpha + 1$, we have that $s(A_i^\alpha \cup \{c_j^{\alpha+2}\}) > B_i$ (otherwise $c_j^{\alpha+2}$ would have been allocated to agent $i$ instead).
        Then, we get that
        \begin{equation*}
            s(A_j^{\alpha + 2}) = s(A_j^{\alpha + 1} \cup \{c_j^{\alpha+2} \})\geq s(A_i^{\alpha} \cup \{c_j^{\alpha+2}\}) > B_i.
        \end{equation*}
        Together with the fact that $A_j^{\alpha+2}$ contains the $\alpha+2$ smallest chores in $A_j$, this means that for any subset $S \subseteq A_j$ with $s(S) \leq B_i$, $|S| \leq |A_j^{\alpha+1}|$, giving us 
        \begin{equation*}
            d(S) = |S| \leq |A_j^{\alpha + 1}| = |A_i^\alpha| + 1 = |A_i| + 1 = d(A_i) + 1.
        \end{equation*}
        This is equivalent to $d(S \setminus \{c\}) \leq d(A_i)$
        for any chore $c \in S$.
    \end{description}
    Finally, we consider the envy by the  housekeeper towards agents in $[n]$.

    Note that every chore $c \in A_{n+1}$ is such that $s(c) \geq s(c')$ for all $c' \in \bigcup_{i=1}^n A_i$.
    Also, for any $i \in [n]$ and $c \in A_{n+1}$ we have $s(A_i \cup \{ c\}) > B_i$.
    
    Suppose for a contradiction there exists a subset $S \subseteq A_{n+1}$ such that $s(S) \leq B_k$ and $|S| - 1 > |A_k|$ for some $k\in [n]$.
    Since $|S| > |A_k| + 1$, consider the subset $S' \subset S$ such that $|S'| = |A_k|$.
    Then, $s(S') \geq s(A_k)$.
    We have that
    \begin{equation*}
        B_k \geq s(S) \geq s(S' \cup \{c\}) \geq s(A_k \cup \{c\}) > B_k
    \end{equation*}
    for some $c \in S \setminus S'$, 
    a contradiction.
    Thus, for any subset $S \subseteq A_{n+1}$ such that $s(S) \leq B_k$, we have that
    \begin{equation*}
        d(S \setminus \{c\}) = |S|-1 \leq |A_k| = d(A_k)
    \end{equation*}
    for any $c \in S$, as desired.
    $\hfill \square$
\end{proof}


\subsection{Identically-Sized Chores}
In the previous subsection, we showed that when chores are identically-valued but have possibly differing sizes, Algorithm \ref{alg:ef1_densestfirst} returns an EF1 allocation.
Now, we show that when chores are identically-sized but with possibly differing disutilities, the algorithm is also able to compute an EF1 allocation. Note that in this case, the algorithm allocates chores with highest disutility first.

\begin{theorem}\label{thm:identicalsize}
    When chores are identically-sized, Algorithm \ref{alg:ef1_densestfirst} returns an \emph{EF1} allocation.
\end{theorem}
\begin{proof}
    Consider any two agents $i, j \in [n]$ with $B_i \leq B_j$ and let $c_i^t$ and $c_j^t$ denote the $t$-th chore added to agent $i$ and $j$'s bundle, respectively.
    Also let $A_i^t$ and $A_j^t$ denote the bundles belonging to agent $i$ and $j$, respectively, after the $t$-th chore was added to their bundles.
Let $c_i^\alpha$ be the last chore that agent $i$ received.

    We first prove the EF1 property by agent $i$ towards agent $j$.
    Since $d(c_i^{t+1}) \leq d(c_j^t)$ for all $t = 1,\dots,\alpha-1$, by summing over $t$ on both sides, we get
        \begin{equation*}
            d(A_i \setminus \{c_i^1 \}) = \sum_{t=1}^{\alpha-1} d(c_i^{t+1}) \leq \sum_{t=1}^{\alpha-1} d(c_j^t) \leq d(A_j).
        \end{equation*}
    Thus, agent $i$ does not envy $j$ by more than one chore. 

    Next, we prove the EF1 property by agent $j$ towards agent $i$. If $|A_j| \leq \alpha+1= |A_i|+1$, by a similar argument as above, we have
    \begin{equation*}
            d(A_j \setminus \{c_i^1 \}) = \sum_{t=1}^{|A_j|-1} d(c_j^{t+1}) \leq \sum_{t=1}^{\alpha} d(c_i^t) \leq d(A_i). 
        \end{equation*}
    Thus, consider the case when $|A_j| > \alpha+1$. Fix any subset $S \subseteq A_j$ such that $s(S) \leq B_i$. Trivially $d(S) \leq d(A_i)$ when $S=\emptyset$, so assume $S \neq \emptyset$. Moreover, we have $|S| \leq \alpha+1$ as otherwise $i$ would have been allocated $\alpha+1$ chores instead.
    Then, since $d(c_j^{t+1}) \leq d(c_j^t)$ for all $t = 1,\dots,\alpha$, by summing over $t$ on both sides, we get
    \begin{equation} \label{eqn:identicalsize_1}
        d(A_j^{\alpha+1} \setminus \{c_j^1 \}) = \sum_{t=1}^\alpha d(c_j^{t+1}) \leq \sum_{t=1}^\alpha d(c_i^t) = d(A_i).
    \end{equation}
    Since $A_j^{\alpha+1} \setminus \{c_j^1 \}$ contains the $\alpha$ chores with highest disutility in $A_j \setminus \{c_j^1 \}$, 
    \begin{equation*}
        d(S \setminus \{c\}) \leq d(A_j^{\alpha+1} \setminus \{ c_j^1 \})
    \end{equation*}
    where $c \in S$ is the chore with highest disutility in $S$.    
    Together with (\ref{eqn:identicalsize_1}) above, we obtain the EF1 property.

    Finally, we prove the EF1 property by the housekeeper towards any agent $k \in [n]$.
    Fix any $S \subseteq A_{n+1}$ with $s(S) \leq B_k$. Note that every chore $c \in A_{n+1}$ satisfies
    \begin{equation} \label{eqn:identicalsize_2}
        d(c) \leq d(c') \quad \text{for all } c' \in \bigcup_{i=1}^n A_i.
    \end{equation}
    Also, for every $c \in A_{n+1}$ we have $s(A_k \cup \{c\}) > B_k $ (otherwise agent $k$ would have been allocated chore $c$). This means that $s(S) - s (c) < s(A_k)$ for every $c \in A_{n+1}$, implying that $|S| \leq |A_k|$ since all the chores have the same size.
Thus, it holds that
    \begin{equation*}
        d(S) \leq \sum_{c' \in A_k} d(c') = d(A_k).
    \end{equation*}
    where the middle inequality follows from (\ref{eqn:identicalsize_2}). 
    $\hfill \square$
\end{proof}

\begin{remark}
We note that for binary disutilities or for identically-sized chores, Algorithm~1 (which computes EFX allocations) runs in polynomial time: indeed, we have argued that the check in Line~3 of that algorithm and computing a feasible set can be reduced to solving Knapsack, and Knapsack is polynomial-time solvable if item sizes or values are polynomially bounded
(and if chores have identical sizes, we can assume without loss of generality that each chore has size $1$). Consequently, in case of binary disutilities or identical sizes, we can compute an EFX allocation in polynomial time. Since every EFX allocation is also an EF1 allocation, this observation can be seen as a strengthening of Theorems~\ref{thm:binary} and~\ref{thm:identicalsize}. 
Nevertheless, Theorems~\ref{thm:binary} and~\ref{thm:identicalsize} remain useful, as they provide guarantees on the performance of a specific natural algorithm, namely, {\sc DensestFirst}.
\end{remark}


\subsection{Identically-Dense Chores}
Next, we consider the case when the chores have identical densities, but with potentially different sizes or disutilities.
We will now show that Algorithm~\ref{alg:ef1_densestfirst} is again able to return an EF1 allocation (the tie-breaking in favor of smaller chores in Line \ref{alg:line:tie} is crucial here). 

\begin{theorem} \label{thm:ef1_identicallydense}
    When chores are identically-dense, Algorithm \ref{alg:ef1_densestfirst} returns an \emph{EF1} allocation.
\end{theorem}
\begin{proof}
    Consider any two agents $i, j \in [n]$ with $B_i \leq B_j$ and let $A_i^t$ and $A_j^t$ denote the bundle belonging to agent $i$ and $j$, respectively, after the $t$-th chore was added to their bundle.
    Since chores are identically-dense, let $\rho$ be the density of each chore. The case when $\rho = 0$ is trivial so we assume that $\rho >0$. 

    We first prove the EF1 property by agent $i$ towards agent $j$.
    If $d(A_i) \leq d(A_j)$, then we are trivially done.
    Hence, we assume that $d(A_i) > d(A_j)$.
    Let there be $\alpha$ chores in $A_i$, and hence $c_i^\alpha$ is the last chore added to $i$'s bundle.
    Then, it must be that $d(A_i \setminus \{c_1^\alpha \})\leq d(A_j)$, otherwise $c_1^\alpha$ would not have been added to agent $i$'s bundle.

    Next, we prove the EF1 property by agent $j$ towards agent $i$.
    If $d(A_j) \leq d(A_i)$, then we are trivially done.
    Hence, we assume that $d(A_j) > d(A_i)$.
    Suppose for a contradiction that for some subset $S \subseteq A_j$ such that $s(S) \leq B_i$, $d(S \setminus \{c\}) > d(A_i)$ for every $c \in S$. 
    Since the density of all the chores is the same, this means that  $s(S \setminus \{c\}) > s(A_i)$ for every $c \in S$. 
    Let $c$ be the last chore allocated to $A_j$ among the chores in $S$ (so $c$ is the largest-sized chore in $S$). 
    Since $s(S \setminus \{c\}) > s(A_i)$, 
    \begin{equation*}
        s(A_i \cup \{c\}) < s(S) \leq B_i,
    \end{equation*}
    and since $d(S \setminus \{c\}) > d(A_i)$, $c$ should have been allocated to agent $i$ instead of agent $j$, a contradiction.
    Therefore, $d(S \setminus \{c\}) \leq d(A_i)$, as desired.
    
    Finally, we prove the EF1 property by the housekeeper towards any agent $k \in [n]$, which is similar to the previous case.
    Suppose for a contradiction that for some subset $S \subseteq A_{n+1}$ such that $s(S) \leq B_k$ it holds that $d(S \setminus \{c\}) > d(A_i)$ for every $c \in S$. Since the density of each chore is the same, this means that $s(S \setminus \{c\}) > s(A_k)$ for every $c \in S$, implying that we have
    \begin{equation*}
        s(A_k \cup \{c\}) < s(S) \leq B_k,
    \end{equation*}
    for every $c \in S$. Thus, at least one chore $c \in S$ should have been allocated to some agent $k' \in [n]$ (by how the algorithm operates), which is a contradiction. Hence, $d(S \setminus \{c\}) \leq d(A_k)$ as desired. 
    $\hfill \square$
\end{proof}

\subsection{Identical Budgets}
In the previous three subsections, we considered the case where chores are identical in some way---be it in value, size, or density.
Now, we relax constraints on these three properties, and consider the case where agents' budgets are identical.
Again, the same Algorithm \ref{alg:ef1_densestfirst} is able to return an EF1 allocation. 

\begin{theorem}\label{thm:ef1_identicalbudgets}
    When agents have identical budgets, Algorithm \ref{alg:ef1_densestfirst} returns an \emph{EF1} allocation.
\end{theorem}
\begin{proof}
Let $B$ be the identical budget. 
    We first show the EF1 property between agents.
    Consider any agent $i \in [n]$.
    We will show that at each iteration of the \textbf{while} loop, the envy agent $i$ has towards any other agent disappears after dropping a new chore from her bundle.
    Note that since agents have identical budgets, any feasible bundle for agent $i$ will be feasible for any other agent $i' \in [n]$ as well.
    Now, the claim is clearly true for the initial allocation $\mathcal{A}$ with $A_i = \emptyset$ for all $i \in [n]$. 
    Assume that the claim holds at some iteration, just before agent $i^*$ is allocated a new chore $c^*$.
    The new allocation that assigns $c^*$ to $i^*$ is EF1 since $i^*$ does not envy any other agent if we remove the chore $c^*$ from her bundle.

    Next, we show the housekeeper is EF1 towards agent $i$.
    Consider any subset of chores $S \subseteq A_{n+1}$ with $s(S) \leq B$.
    Let $c \in S$ be a chore with maximum density among the chores in $S$, i.e., 
    $c \in \argmax_{c' \in S} \rho(c')$.
    Also, let $W_c := \{ h \in A_i \mid \rho(h) \geq \rho(c)\}$.
    Note that $W_c \neq \emptyset$; otherwise, $A_i = \emptyset$ and chore $c$ with $s(c) \leq s(S) \leq B$ would have been added to $i$'s bundle.

    By a similar reason, we have that $s(W_c \cup \{c\}) > B$.
    Since $s(S \setminus \{c\}) \leq B - s(c)$, this means $s(W_c) > B  - s(c) \geq s(S \setminus \{c\})$.
    Thus, we get that $d(A_i) \geq d(W_c) \geq s(W_c) \times \rho(c) > s(S \setminus \{c\}) \times \rho(c) \geq d(S \setminus \{c\})$, as desired.
    $\hfill \square$
\end{proof}

\subsection{Two Agents}
The last special case that we consider is the setting with two agents, which is often studied in the fair allocation literature
\cite{BramsTa96,Budish11,CaragiannisKuMo19,GanLiWu2023,GarbeaGkTa2023,igarashi2023repeated}.
In fact, the case of two agents is particularly important in the allocation of chores, given that a key application of our results is the domain of household chore division, and many households consist of two adults~\cite{igarashi2023choreapp}.

We propose an algorithm (Algorithm \ref{alg:ef1_twoagents}, which is similar to Algorithm 4 in~\cite{GanLiWu2023}) that returns an EF1 allocation. The algorithm uses Algorithm \ref{alg:ef1_densestfirst} as a subroutine and runs in polynomial time.

\begin{algorithm}
\caption{Computes an EF1 allocation for two agents}
\label{alg:ef1_twoagents}
\SetInd{0.8em}{0.3em}
\textbf{Input }disutility function $d$, size function $s$, budgets $(B_1, B_2)$ with $B_1\le B_2$\;
Run Algorithm~\ref{alg:ef1_densestfirst} on both agents with identical budget $B_1$ and obtain allocation $\mathcal{A}' = (A'_1, A'_2,A'_3)$\;
\If{$d(A'_1) \geq d(A'_2)$}{
    $A_1 \leftarrow A'_1$\;
}
\Else{
    $A_1 \leftarrow A'_2$\;
}
Run Algorithm~\ref{alg:ef1_densestfirst} on agent $2$ with budget $B_2$ and set of chores $C \setminus A_1$. Let the output be $(A_2, A_3)$\;
\Return allocation $(A_1, A_2, A_3)$\;
\end{algorithm}

\begin{theorem} \label{thm:ef1_twoagents}
    When there are two agents, Algorithm \ref{alg:ef1_twoagents} returns an \emph{EF1} allocation in polynomial time.
\end{theorem}
\begin{proof}
    The fact that Algorithm \ref{alg:ef1_twoagents} runs in polynomial time is easy to observe, given that Algorithm \ref{alg:ef1_densestfirst} runs in polynomial time (as is proven in Theorem \ref{thm:densestfirst_runtime}), and the other operations take polynomial time as well.

    Next, we prove the correctness of the algorithm.
    Without loss of generality, suppose that $A_1=A'_1$.
    Note that $B_1 \leq B_2$.
    We first show that agent 1 does not envy agent 2 after dropping a chore from his own bundle.

    Observe that $A'_2$ is produced by running Algorithm \ref{alg:ef1_densestfirst} on a single agent (agent 2) with budget $B_2$ on items $C \setminus A_1$.
    If $A_2\setminus A'_2=\emptyset$, then
    $A_2 = A'_2$ (since $B_1 \leq B_2$) and the
    result follows trivially, so assume that this is not the case.
    Let $c^*$ be the first chore in $A_2\setminus A'_2$ allocated to agent 2.
    Let $X$ be agent 2's bundle right before the algorithm allocates $c^*$.
    Then, since Algorithm \ref{alg:ef1_densestfirst} allocates chores in a densest-first fashion, it must be that $X \subseteq A'_2$.
    Since $c^* \in A_2$ and $c^* \notin A'_2$, we have 
\begin{equation*}
        s(A'_2) \leq B_1 < s(X \cup \{c^*\}) \leq B_2.
    \end{equation*}
    Since chores in $A_2'\setminus X$ have density at most that of $c^*$ and $X \subseteq A'_2$, we get 
    $d(X \cup \{c^*\}) > d(A_2')$, and
    \begin{equation*}
        d(A_2) \geq d(X \cup \{c^*\}) > d(A_2').
    \end{equation*}
    Then, since $d(A_1' \setminus \{c\}) \leq d(A_2')$, where $c$ is the last chore added to $A'_1$, we get that
    \begin{equation*}
        d(A_1 \setminus\{c\}) = d(A'_1 \setminus\{c\}) \leq d(A'_2) < d(A_2),
    \end{equation*}
    as desired.

    Next, we show that agent 2 as well as the housekeeper do not envy agent 1 by more than one chore. To prove this, consider any subset $S \subseteq A_2 \cup A_3$ with $s(S) \leq B_1$. If $S \subseteq A_2'$, then $d(S) \leq d(A'_2) \leq d(A'_1)=d(A_1)$. Thus, assume that $S \setminus A_2' \neq \emptyset$. Let $c \in S \setminus A_2'$ be a chore with maximum density $\rho_c$ among the chores in $S \setminus A_2'$.
    
    Let $W_c = \{\, j \in A'_2 \mid \rho_j \geq \rho_c \,\}$. Note that $W_c \neq \emptyset$; otherwise, $A'_2=\emptyset$ and chore $c$ with $s(c) \leq s(S) \leq B_1$ would have been added to $A'_2$. 
    Since $c$ is not included in $A'_2$, we have $s(W_c \cup \{c\}) > B_1$. Moreover, since $s(S \setminus \{c\}) \leq B_1-s(c)$, this means that
    $$
    s(W_c)> B_1-s(c) \geq s(S \setminus \{c\}). 
    $$
    Thus, we get 
    \[
    d(A_1) \geq d(A_2') \geq d(W_c) \geq s(W_c) \times \rho_c > s(S \setminus \{c\}) \times \rho_c  \geq d(S \setminus \{c\}).
    \]
Finally, the EF1 property by the housekeeper towards agent 2 can be easily verified due to Theorem~\ref{thm:ef1_identicalbudgets}.  
    $\hfill \square$
\end{proof}

\section{Computing EF Allocations for Divisible Chores}
We will now shift our focus to divisible chores.
Recently, Barman et al.~\cite{BarmanKhShSr2023} showed the existence of EF allocations for divisible goods under generalized assignment constraints (refer to Section 2).
We show that the same result holds for chores in the budget-constrained setting.\footnote{Again, we only focus on the identical-size setting due to the unintuitive nature of definitions for chores under generalized assignment constraints; refer to Section \ref{sec:preliminaries}.}
As the proof of the result is considerably involved and in light of space constraints, we defer it to the appendix. 
\begin{theorem} \label{thm:divisible_ef}
    An \emph{EF} allocation for divisible chores exists, and can be computed in polynomial time.
\end{theorem}

\section{Conclusion}
In this work, we propose a model of allocating indivisible or divisible chores under budget constraints.
For indivisible chores, we prove the existence of EFX allocations. Our proof is constructive and provides a pseudopolynomial time algorithm for finding EFX allocations.
Moreover, we put forward a polynomial-time algorithm that returns an EF2 allocation for general instances, and EF1 allocations in five special cases---when chores are identically-valued, identically-sized, identically-dense, when agents have identical budgets, and the case of two agents.

For divisible chores, we show the existence of polynomial-time computable EF allocations in general.

Possible future directions include exploring definitions of envy-freeness when agents have subjective size or disutility functions (i.e., generalized assignment constraints for chores), or considering approximate EF guarantees under non-additive size or disutility functions (while maintaining the identical size/disutility function assumption). 
Other avenues for future work include considering an online model \cite{AAGW15,BKP+24,elkind2024temporalfairdivision}, or looking into the more general class of \emph{submodular} valuations \cite{ghodsi2022submodular,montanari2024submodular,SuksompongTe23,uziahu2023submodular} for the budget-constrained setting.
Another exciting direction is to extend our formal model and results to mixed manna, i.e., items that are viewed as goods by some agents and as chores by others \cite{AzizCaIg22chores};
anecdotally, this model may be appropriate for some household tasks, such as cooking, gardening, or spending time with children or animals.

\bibliographystyle{splncs04}
\bibliography{abb,bib}

\newpage 

\appendix

\begin{center}
\Large
\textbf{Appendix}
\end{center}

\vspace{2mm}

\section{Proof of Theorem \ref{thm:ef2}}
We first introduce some useful notation.
For any subset of chores $S = \{c_1, \dots, c_k\}$ indexed in decreasing order of density, let $S^{(i)}$ be the subset of the $i$ densest chores in $S$.
Moreover, for any nonnegative number $B \leq s(S)$, let $P = \{c_1,\dots,c_{\ell-1} \}$ be the (cardinality-wise) largest prefix of $S$ such that $s(P) \leq B$.
Then we define $S^{[B]} := P \cup \{ \alpha \cdot c_\ell\}$ where $\alpha = \frac{B - s(P)}{s(c_\ell)}$.
If $B \geq s(S)$, then $S^{[B]} = S$.

Consider the following two observations which immediately follow from the construction of the algorithm. Here, let $A_a = \{c_1,\dots,c_k\}$ and $A_b = \{e_1,\dots, e_\ell\}$ denote the sets of chores assigned to agents $a$ and $b$ respectively.
        Then at the end of the algorithm, the chores in each of these sets are indexed in non-increasing order of density.
\begin{description}
    \item[Observation 1.] 
        For indices $i < |A_a|$ and $j < |A_b|$, if $d(A_a^{(i)}) < d(A_b^{(j)})$ and $s(A_a^{(i)} + e_{j+1}) \leq B_a$, then $\rho(c_{i+1}) > \rho(e_{j+1})$.
    \item[Observation 2.] 
        If, for any index $j \in [|A_b| - 1]$, the size $s(A_a + e_{j+1}) \leq B_a$, then $d(A_a) \geq d(A_b^{(j)})$.
\end{description}
    

    Next, we define the following function:
    \begin{equation}
        \textsc{EFCount}(X,Y) := \min_{R \subseteq X : d(X \setminus R) \leq d(Y)} |R|.
    \end{equation}
    Intuitively, the function returns the number of chores that need to be removed from $X$ in order to achieve envy-freeness by $X$ towards $Y$.

    \begin{lemma} \label{lem:efcount_increment1}
        For any subset of chores $X$ and $Y$ along with any index $i < |X|$, let $T = s(X^{(i)})$ and $\widehat{T} = s(X^{(i+1)})$.
        Then, 
        $$
        \textsc{EFCount}(X^{[\widehat{T}]},Y^{[\widehat{T}]}) \leq \textsc{EFCount}(X^{[T]},Y^{[T]})+1.
        $$
    \end{lemma}
    \begin{proof}
        Let $\ell = \textsc{EFCount}(X^{[T]}, Y^{[T]})$.
        Therefore, by definition, there exists a size-$\ell$ subset $R \subseteq X^{[T]}$ with the property that $d(X^{[T]} \setminus R) \leq d(Y^{[T]})$.

        Define subset $R' = R \cup \{ e_{i+1}\}$ where $e_{i+1}$ is the chore in the set $X^{(i+1)} \setminus X^{(i)}$.
        For this set $R'$ whereby $|R'| = \ell+1$, we have
        \begin{equation*}
            d(X^{[\widehat{T}]} \setminus R') = d(X^{[T]} \setminus R) \leq d(Y^{[T]}) \leq d(Y^{[\widehat{T}]}),
        \end{equation*}
        which implies $\textsc{EFCount}(X^{[\widehat{T}]}, Y^{[\widehat{T}]}) \leq \ell+1$, giving us the result as desired.
        $\hfill \square$
    \end{proof}

    \begin{lemma} \label{lem:envycount_exactly2}
        Let $X$ and $Y$ be any subset of chores with the property that 
        $$
        \textsc{EFCount}(X,Y) \geq 2.
        $$
        Then, there exists an index $t \leq |X|$ such that with $T = s(X^{(t)})$ it holds that $\textsc{EFCount}(X^{[T]},Y^{[T]}) = 2$.
    \end{lemma}
    \begin{proof}
        The lemma follows from a discrete version of the Intermediate Value Theorem.
        For indices $t \in \{0,1,2,\dots, |X|\}$, define the function $h(t) := s(X^{(t)})$, i.e., $h(t)$ denotes the size of the $t$ densest chores in $X$.
        Extending this function, we consider the envy count at different size thresholds; in particular, write 
        \begin{equation*}
            H(t) := \textsc{EFCount}(X^{[h(t)]}, Y^{[h(t)]})
        \end{equation*}
        for each $t \in \{0,1,2,\dots, |X|\}$. Note that $H(0) = 0$. We will show that
        \begin{enumerate}[(i)]
            \item $H(|X|) \geq 2$, and
            \item the discrete derivative of $H$ is at most one, i.e., $H(t+1)-H(t) \leq 1$ for all $0 \leq t < |X|$.
        \end{enumerate}
    These properties of the integer-valued function $H$ implies that there necessarily exists an index $t^*$ such that $H(t^*) = 2$.
    This index satisfies the lemma.
    Therefore, we complete the proof by establishing properties (i) and (ii) for the function $H(\cdot)$.

    For the property (i), note that the definition of a prefix subset gives us
    \begin{equation*}
        d(Y) \geq d(Y^{[s(X)]}).
    \end{equation*}
    Hence, $\textsc{EFCount}(X,Y^{[s(X)]}) \geq \textsc{EFCount}(X,Y) \geq 2$ (from the assumption of this lemma).
    Since $h(|X|) = s(X)$, we have $H(|X|) \geq 2$.

    Property (ii) follows directly from Lemma \ref{lem:efcount_increment1}, completing the proof.
    $\hfill \square$
    \end{proof}

Intuitively, the following lemma states that if one adds more disutility to an agent that envies another by exactly $\ell \geq 1$ chores, then the $\textsc{EFCount}$ does not increase.
\begin{lemma} \label{lem:efcount_mainlemma}
    Let $Z$ and $Y$ be sets of chore along with two nonnegative thresholds $T, \widehat{T} \in \mathbb{R}_+$ and a positive integer $\ell$ with properties
    \begin{enumerate}[(i)]
        \item $\textsc{EFCount}(Z^{[\widehat{T}]}, Y^{[T]}) = \ell$, and
        \item $d(Z \setminus Z^{[\widehat{T}]}) \leq d(Y \setminus Y^{[T]})$. 
    \end{enumerate}
    Then $\textsc{EFCount}(Z,Y) \leq \textsc{EFCount}(Z[\widehat{T}],Y[T]) = \ell$.
\end{lemma}
\begin{proof}
    Given that $\textsc{EFCount}(Z^{[\widehat{T}]}, Y^{[T]}) = \ell$, there exist $\ell$ chores $c'_1, c'_2,\ldots,c'_{\ell} \in Z^{[\widehat{T}]}$ such that 
    \begin{equation*}
        d(Z^{[\widehat{T}]} \setminus \{c'_1,c'_2,\ldots,c'_{\ell}\}) \leq d(Y^{[T]}).
    \end{equation*}
    Using the definition of prefix subsets, we get
    \begin{align*}
        d(Y) 
        & = d(Y^{[T]}) + d(Y \setminus Y^{[T]}) \\
        & \geq d(Z^{[\widehat{T}]} \setminus \{c'_1,c'_2,\ldots,c'_{\ell}\}) + d(Y \setminus Y^{[T]}) \\
        & \geq d(Z^{[\widehat{T}]} \setminus \{c'_1,c'_2,\ldots,c'_{\ell}\}) + d(Z \setminus Z^{[\widehat{T}]}) \\
        & = d(Z) - \sum^{\ell}_{h=1}d(c'_h),
    \end{align*}
    where the third line follows by assumption (ii) of the lemma.
    The definition of the prefix subset $Z^{[\widehat{T}]}$ ensures that, corresponding to chores $c'_1, c'_2,\ldots,c'_{\ell} \in Z^{[\widehat{T}]}$, there exists $\ell$ chores $c_1, c_2,\ldots,c_{\ell} \in Z$ such that
    \begin{equation*}
        \sum^{\ell}_{h=1}d(c_h) \geq \sum^{\ell}_{h=1}d(c'_h).
    \end{equation*}
    This bound and the above derived inequality gives us
    $d(Z \setminus \{c_1,c_2,\ldots,c_{\ell}\}) \leq d(Y).$
    This implies $\textsc{EFCount}(Z^{[\widehat{T}]}, Y^{[T]}) \leq \ell$ and completes the proof of the lemma.
    $\hfill \square$
\end{proof}
Note that the above lemma holds in particular when $\textsc{EFCount}(Z^{[\widehat{T}]},Y^{[T]}) = 2$ for nonnegative thresholds $T, \widehat{T} \in \mathbb{R}_+$.

We now proceed to prove the main result of this theorem.
    Fix any two agents $a,b \in [n+1]$ and let $A_a$ and $A_b$ be the sets of chores allocated to them, respectively, at the end of the algorithm. Let $A_a = X = \{x_1,\dots,x_k\}$ and $A_b = Y = \{y_1,\dots,y_\ell\}$ for ease of notation. Note that chores are indexed in non-increasing order of density. Proving EF2 between the two agents $a,b$ corresponds to showing that, for any subset of chores $Z \subseteq X$ with $s(Z) \leq B_b$, we have 
    \begin{equation*}
        \textsc{EFCount}(Z,Y) \leq 2.
    \end{equation*}
    Consider any such subset $Z$ and index its chores in non-increasing order of density $Z = \{z_1,\dots,z_{\ell'}\}$.
    Note that if $\textsc{EFCount}(Z,Y) \leq 1$, we already have EF2.
    Thus, the remainder of the proof, we consider the case where
    \begin{equation}
        \textsc{EFCount}(Z,Y) \geq 2.
    \end{equation}
    We show that the inequality cannot be strict, i.e., it must hold that envy count is at most $2$, and hence EF2.

    Start by considering the function $h(i) = s(Z^{(i)})$ for $i \in \{0,1,2,\dots,|Z|\}$.
    Furthermore, define index
    \begin{equation} \label{eqn:efcount_tmin}
        t = \min \{ i: \textsc{EFCount}(Z^{(i)},Y^{[h(i)]})  = 2\}
    \end{equation}
    The existence of such an index $t \geq 2$ follows from Lemma \ref{lem:envycount_exactly2}.
    Also, note that $Z^{(i)} = Z^{[h(i)]}$.
    Denote the $t$-th chore in $Z$ by $c_Z = z_t$.
    In addition, using $t$, define the following two size thresholds:
    \begin{equation*}
        \tau = s(Z^{(t-1)}) \quad \text{and} \quad \widehat{\tau} = s(Z^{(t)}).
    \end{equation*}

    Now, Lemma \ref{lem:efcount_increment1}, the definition of $t$, and Equation (\ref{eqn:efcount_tmin}) give us $$
    \textsc{EFCount}(Z^{[\tau]}, Y^{[\tau]}) \geq 1.
    $$
    Using the minimality of $t$, 
    \begin{equation} \label{eqn:efcount_equals1}
        \textsc{EFCount}(Z^{[\tau]}, Y^{[\tau]}) = 1.
    \end{equation}

    We will establish two properties for sets $Z$ and $Y$ under consideration, and show
    \begin{enumerate}[(i)]
        \item $\textsc{EFCount}(Z^{[\widehat{\tau}]}, Y^{[\tau]}) = 2$, and
        \item $d(Z \setminus Z^{[\widehat{\tau}]}) \leq d(Y \setminus Y^{[\tau]})$,
    \end{enumerate}
    and apply Lemma \ref{lem:efcount_mainlemma} to get the desired result.

    We prove the first property with the following Lemma.
    \begin{lemma}
        $\textsc{EFCount}(Z^{[\widehat{\tau}]}, Y^{[\tau]}) = 2$.
    \end{lemma}
    \begin{proof}
        Since $\textsc{EFCount}(Z^{[\tau]},Y^{[\tau]}) = 1$ (by \eqref{eqn:efcount_equals1}), there exists a chore $c_1 \in Z^{[\tau]}$ such that
        \begin{equation} \label{eqn:efcount_lemma_chore}
            d(Z^{[\tau]} \setminus \{ c_1 \}) \leq d(Y^{[\tau]}).
        \end{equation}
        Also, by definition
        \begin{equation*}
            Z^{[\widehat{\tau}]} = Z^{[\tau]} \cup \{ c_Z\}. 
        \end{equation*}
        So
        \begin{equation*}
            Z^{[\tau]} = Z^{[\widehat{\tau}]} \setminus \{ c_Z\}
        \end{equation*}
        and Equation (\ref{eqn:efcount_lemma_chore}) reduces to 
        \begin{equation*}
            d(Z^{[\widehat{\tau}]} \setminus \{ c_Z , c_1 \} \leq d(Y^{[\tau]})
        \end{equation*}
        Therefore, $\textsc{EFCount}(Z^{[\widehat{\tau}]}, Y^{[\tau]}) = 2$.
        $\hfill \square$
    \end{proof}

    Next, we define $\gamma$ to be the size of the chores in $Y$ that are at least as dense as $c_Z$, i.e.,
    \begin{equation*}
        \gamma = \sum_{c \in Y : \rho(c) \geq \rho(c_Z)} s(c).
    \end{equation*}
    \begin{lemma} \label{lem:efcount_twoproperties}
        It holds that (1) $\gamma \leq \widehat{\tau}$ and (2) $d(Z^{[\tau]}) > d(Y^{[\gamma]})$.
    \end{lemma}
    \begin{proof}
        We first establish property (1), the upper bound on $\gamma$.

        Assume for a contradiction that $\gamma > \widehat{\tau}$.
        By definition of $\gamma$, we have that all chores in $Y^{[\widehat{\tau}]}$ is at least $\rho(c_Z)$.
        In particular, all the chores in the set $Y^{[\widehat{\tau}]} \setminus Y^{[\tau]}$ are at least as dense as $c_Z$.
        Thus, 
        \begin{equation}
            d(Y^{[\widehat{\tau}]} \setminus Y^{[\tau]}) \geq d(Z^{[\widehat{\tau}]} \setminus Z^{[\tau]}) = d(c_Z).
        \end{equation}
        Together with the fact that $\textsc{EFCount}(Z^{[\tau]}, Y^{[\tau]}) = 1$ (from Equation (\ref{eqn:efcount_equals1})),  we have $\textsc{EFCount}(Z^{[\widehat{\tau}]}, Y^{[\widehat{\tau}]}) \leq 1$ (by Lemma \ref{lem:efcount_mainlemma}), contradicting the definition of $t$ (and correspondingly, $\widehat{\tau}$. 
        Thus, claim (1) holds.

        Next, we prove property (2).
        Suppose for a contradiction that $d(Z^{[\tau]}) \leq d(Y^{[\gamma]})$.
        Since $\gamma \leq \widehat{\tau}$, we get that 
        \begin{equation*}
            d(Z^{[\widehat{\tau}]} - c_Z) = d(Z^{[\tau]}) \leq d(Y^{[\gamma]}) \leq d(Y^{[\widehat{\tau}]})
        \end{equation*}
        This implies that $\textsc{EFCount}(Z^{[\widehat{\tau}]}, Y^{[\widehat{\tau}]}) \leq 1$.
        This envy count contradicts the definition of $t$ (and correspondingly $\widehat{\tau}$).
        Thus, the claim holds.
        $\hfill \square$
    \end{proof}

    We prove the second property with the following Lemma.
    \begin{lemma}
        $d(Z \setminus Z^{[\widehat{\tau}]}) \leq d(Y \setminus Y^{[\tau]})$.
    \end{lemma}
    \begin{proof}
        Since $Z \subseteq X$, the chore $c_Z$ appears in the subset $X$.
        Recall that the chores in the subsets $Z$ and $X = \{x_1,\dots,x_k\}$ are in order of decreasing density.
        Write $t' \in [|X|]$ to denote the index of $c_Z$ in $X$, i.e., $c_Z = x_{t'}$.

        Then, the second property of Lemma \ref{lem:efcount_twoproperties} gives us $d(Z^{[\tau]}) > d(Y^{[\gamma]})$, which in turn gives us
        \begin{equation*}
            \sum_{i=1}^{t'-1} d(x_i) = \sum_{i=1}^{t-1} d(z_i) = d(Z^{[\tau]}) > d(Y^{[\gamma]}).
        \end{equation*}
        That is, $d(X^{(t'-1)}) > d(Y^{[\gamma]})$.
        By definition of $\gamma$, we have that chores in $Y \setminus Y^{[\gamma]}$ (if any) have density less than $\rho(c_Z)$.
        These combined with Observation 1 gets that including $c_Z$ in $Y^{[\gamma]}$ must violate agent $b$'s budget $B_b$.
        We get that 
        \begin{equation*}
            \gamma + s(c_Z) > B_b.
        \end{equation*}

        Note that 
        \begin{align*}
            s(Y^{[\gamma]} \setminus Y^{[\tau]}) 
            &= \gamma - \tau\\
            & = \gamma - \widehat{\tau} + s(c_Z)\\
            & > B-b - \widehat{\tau}\\
            & \geq s(Z \setminus Z^{[\tau]}),
        \end{align*}
        where the first line follows by definition of the respective prefix subsets, the second line follows from the fact that $Z^{[\widehat{\tau}]} = Z^{[\tau]} \cup \{c_Z\} \text{ and } \widehat{\tau} = \tau + s(C_z)$, and third line follows from the fact that $s(Z) \leq B_b$.
        
        By definition of $\gamma$ ($\gamma$ is the size of chores in $Y$ that is at least as dense as $c_Z$), we have that every chore $c \in Y^{[\gamma]} \setminus Y^{[\tau]}$ has density $\rho(c) \geq \rho(c_Z)$.

        In addition, for every chore $c' \in Z \setminus Z^{[\widehat{\tau}]}$, $\rho(c') \leq \rho(c_Z)$.
        This gives us
        \begin{equation*}
            \rho(c) \geq \rho(c_Z) \geq \rho(c'),
        \end{equation*}
        for every chore $c \in Y^{[\gamma]} \setminus Y^{[\tau]}$ and every chore $c' \in Z \setminus Z^{[\widehat{\tau}]}$.

        Together with the fact that $s(Z \setminus Z^{[\widehat{\tau}]}) \leq s(Y^{[\gamma]} \setminus Y^{[\tau]})$, we get
        \begin{equation*}
            d(Z \setminus Z^{[\widehat{\tau}]}) \leq d(Y^{[\gamma]} \setminus Y^{[\tau]}),
        \end{equation*}
        because the RHS is more dense and bigger size.
        Then, since $Y^{[\gamma]} \subseteq Y$, we get \begin{equation*}
            d(Z \setminus Z^{[\widehat{\tau}]}) \leq d(Y \setminus Y^{[\tau]}),
        \end{equation*}
        as desired.
        $\hfill \square$
    \end{proof}
    
    Applying Lemma \ref{lem:efcount_mainlemma}, we get the desired property (including housekeeper condition).

\section{Proof of Theorem \ref{thm:divisible_ef}}
We will show the existence of an algorithm that returns an EF  allocation for divisible chores.
The argument is similar to that of the goods case \cite{BarmanKhShSr2023}.
Note that the result holds even when disutility functions can be subjective---we use $d_i$ to denote the disutility function of agent $i$.\footnote{Note that the housekeeper envy in this case will be defined using each agents' disutility function, which may be highly unintuitive (as noted in Section 2).}

We first introduce some key ingredients for the algorithm and its proof.
For any two bundles (which are also vectors in $[0,1]^m$) $x, y$, let $x \cap y \in [0,1]^m$ denote the vector whose $i$-th component is $\min\{x_i, y_i\}$ for all $i \in [m]$.
Moreover, we let $x - y \in [0,1]^m$ denote the vector whose $i$-th component is $\max\{0,x_i - y_i\}$ for all $i \in [m]$.

Given any instance, we include a \emph{fractional} chore $c_{m+1}$ that has disutility $d_i(c_{m+1}) = 0$  for all agents $i \in [n]$ and size $s(c_{m+1}) = 2n \cdot \max_{{i'} \in [n]} B_{i'}$.
The inclusion of this chore, $c_{m+1}$, ensures that we can always work with allocations $\mathcal{X} = (X_1,\dots, X_n) \in [0,1]^{n\times (m+1)}$ where  the size of each agents' bundle is exactly his budget, i.e., $s(X_i) = B_i$.
This equality can be obtained by choosing an appropriate $x_{i,m+1} \in [0,\frac{1}{2n}]$.

The inclusion of the $(m+1)$-th chore also implies that the fractional assignment to each agent $i$ is denoted by an $(m+1)$-dimensional vector $X_i = (x_{i,1},x_{i,2},\dots,x_{i,m},x_{i,m+1})$, where $x_{i,j}$ denotes the fraction of chore $j$ assigned to $j$.

We show, with the following proposition, that including this new chore $c_{m+1}$ does not affect the EF property.
\begin{proposition} \label{prop:divisible_efhcremove}
    Suppose that $\mathcal{X} = (X_1,\dots,X_n) \in [0,1]^{n \times (m+1)}$ is an \emph{EF} allocation for the instance with $m+1$ chores. Then there exists an \emph{EF} allocation $\overline{\mathcal{X}}= (\overline{X}_1,\dots,\overline{X}_n) \in [0,1]^{n \times m}$ for the instance with $m$ chores. 
\end{proposition}
\begin{proof}
    The allocation $\overline{\mathcal{X}} = (\overline{X}_1,\dots,\overline{X}_n)$ is obtained from $\mathbf{\mathcal{X}} = (X_1,\dots,X_n)$ by removing the fictional chore $c_{m+1}$ from consideration.
    Recall that all agents $i \in [n]$ have zero disutility for the chore $c_{m+1}$, hence  the disutility under the two allocations remains unchanged: $d_i(\overline{X}_i) = d_i(X_i)$.

    We will show that $\overline{X}$ is also EF allocation.
    For any agent $i \in [n]$,
    consider any relevant (feasible) fractional assignment $\overline{Y} \in [0,1]^m$ for $i$ such that $\overline{Y} \leq \overline{X}_i$ or $\overline{Y} \leq \overline{X}_{n+1}$.
    Extend $\overline{Y}$ to obtain the fractional assignment $Y \in [0,1]^{m+1}$ as follows:
    \begin{itemize}
        \item Set $y_j = \overline{y}_j$ for all $j \in [m]$, along with $y_{m+1} = 0$.
        \item Note that the fractional assignment $Y \in [0,1]^{m+1}$ is considered when applying definition EF for allocation $\mathcal{X} \in [0,1]^{n \times (m+1)}$ and the fact that $\mathcal{X}$ is EF implies $d_i(Y) \leq d_i(X_j)$ and $d_i(Y) \leq d_i(X_i)$ respectively.
        This means $d_i(\overline{Y}) \leq d_i(\overline{X}_j)$ and $d_i(\overline{Y}) \leq d_i(\overline{X}_i)$ respectively.
    \end{itemize}
    Since this inequality holds for all agents and relevant fractional assignments $\overline{Y} \in [0,1]^m$, $\overline{\mathcal{X}}$ is an EF allocation.
    $\hfill \square$
\end{proof}

Now, recall that for any chore $c$ and agent $i$, 
\begin{equation*}
    \rho_i(c) = \frac{d_i(c)}{s(c)}.
\end{equation*}

We will define a property called \emph{density domination}, and prove that any fractional allocation that satisfies this property is EF.

For each agent $i \in [n]$, let  $\pi_i : [m+1] \rightarrow [m+1]$ be the density ordering across the chores.
\begin{itemize}
    \item $\pi_i(t)$ denotes the $t$-th most dense chore according to $\rho_i(\cdot)$, for each index $1 \leq t \leq m+1$.
    \item If two chores have the same density for $i$, break ties according to the original indexing of the chores, so either:
    \begin{itemize}
        \item $\rho_i(\pi_i(t)) > \rho_i(\pi_i(t+1))$, or
        \item $\rho_i(\pi_i(t)) = \rho_i(\pi_i(t+1))$ and $\pi_i(t) < \pi_i(t+1)$.
    \end{itemize}
\end{itemize}
Also, for each agent $i \in [n]$, we have $\pi_i(m+1) = m+1$.

Then, we define two sets of chores, as follows.
\begin{definition}[Internal and Edge chores]
    For any integer vector $\tau = (\tau_1,\tau_2,\dots \allowbreak \tau_n) \in \mathbb{Z}^n_+$ with $||\tau||_\infty \leq m+2$, and for any agent $i \in [n]$, the set of \emph{internal chores}, $I_i(\tau)$ is defined as the set of the $(\tau_i - 1)$ most dense chores for agent $i$, i.e.,
\begin{equation*}
    I_i(\tau) = \{ \pi_i(1), \pi_i(2),\dots, \pi_i(\tau_i - 1)\}
\end{equation*}
In addition, for agent $i$, the set of \emph{edge chores} $E_i(\tau)$ is defined as
\begin{equation*}
    E_i(\tau) = 
    \begin{cases} 
      \{\pi_i(\tau_i)\} & \text{ if } \tau_i \leq m+1 \\
      \emptyset & \text{ otherwise (if } \tau_i = m+2) 
   \end{cases}
\end{equation*}
where
\begin{itemize}
    \item $I(\tau) = \bigcup_{i=1}^n I_i(\tau)$ is the set of internal chores
    \item $E(\tau) = \bigcup_{i=1}^n E_i(\tau)$ is the set of edge chores
\end{itemize}
$\tau_i = m+2$ denotes that for agent $i \in [n]$, all the chores are internal: $I_i(\tau) = [m+1]$ and the edge set for agent $i$ is empty: $E_i(\tau) = \emptyset$.
If $\tau_i = 1$, then the set of internal chores $I_i(\tau)$ is empty.
\end{definition}

Then, we define the density domintation property, which will be key in showing the EF property we desire.

\begin{definition}[Density Domination] \label{def:dd}
    An allocation $\mathcal{X} = (X_1,\dots,X_n)$ is said to satisfy the \emph{density domination (DD)} property if and only if there exists an integer vector $\widehat{\tau} \in \mathbb{Z}^n_+$ (with $||\widehat{\tau}||_\infty \leq m+2$) such that for all agents $i,i' \in [n]$, we have
    \begin{enumerate}[(i)]
        \item $x_{i,j} \leq x_{i',j}$ for all chores $c_j \in I_i(\widehat{\tau})$, i.e., if a chore $c_j$ is internal to an agent $i$, then the fraction of $c_j$ is assigned to $i$ is at most the fraction of the chore assigned to any other agent;
        \item $\sum_{j:c_j \in I_i(\widehat{\tau}) \cup E_i(\widehat{\tau})} x_{i,j} \cdot s(c_j) = B_i$, i.e., under allocation $\mathcal{X}$, for every agent $i$, the budget constraint is satisfied with an equality; and
        \item $\sum_{i = 1}^n x_{i,j} = 1$ for all chores $c_j \in I(\widehat{\tau})$, i.e., if a chore $c_j$ is internal to any agent, then no fraction of $c_j$ is left unassigned (i.e., none of it goes to the housekeeper), i.e., it is entirely divided among all the agents in allocation $\mathcal{X}$.
    \end{enumerate}
\end{definition}
Note that (i) implies that if a chore $c_j$ is internal to two agents ${i}$ and ${i'}$, then the fractions of $c_j$ that ${i}$ and ${i'}$ receives must be exactly equal, $x_{{i},j} = x_{{i'},j}$.

Furthermore, following from (ii), in an allocation that satisfies the DD property, each agent $i$ is allocated to fractions of only the top $\widehat{\tau}_i$ densest chores according to him, i.e.,
\begin{equation*}
    I_i(\widehat{\tau}) \subseteq supp(X_i) \subseteq I_i(\widehat{\tau}) \cup E_i(\widehat{\tau}),
\end{equation*}
where $supp(X_i)$ denotes the subset of chores that are fractionally assigned under vector $X_i$. 

We prove that DD implies EF in the following lemma.
\begin{lemma} \label{prop:divisible_efhc}
    Any allocation $\mathcal{X} = (X_1, \dots, X_n)$ satisfying the \emph{DD} property is \emph{EF}.
\end{lemma}
\begin{proof}
    We first show the EF property between agents.
    Consider any two agents $i,i' \in [n]$ and any fractional assignment
    \begin{equation} \label{eqn:divisible_yleqa}
        Y \leq X_i
    \end{equation}
    such that 
    \begin{equation} \label{eqn:divisible_sayxb}
        s(Y) \leq B_{i'}.
    \end{equation}
    To prove that $d_i(Y) \leq d_i(X_{i'})$, we consider the fractional assignments $(X_{i'} - Y) \in [0,1]^{m+1}$ and $(Y-X_{i'}) \in [0,1]^{m+1}$ (recall the vector operations detailed in Section \ref{sec:preliminaries}).
    Additionally, denote $\widehat{\tau} = (\widehat{\tau_1}, \dots, \widehat{\tau_n})$ as the integer vector that certifies the DD property of $\mathcal{X}$.

    Note that the DD implies that for all internal chores $c_j \in I_i(\widehat{\tau})$ we have
    \begin{equation*}
        y_j \leq x_{i,j} \leq x_{i',j},
    \end{equation*}
    where the left inequality is by (\ref{eqn:divisible_yleqa}) and the right inequality is by condition (ii) of the DD definition.

    Using this inequality and the definition of $I_i(\widehat{\tau})$, we know that all the chores in the set $supp(Y-X_{i'}) = \{j' \in [m+1]: y_{j'} > x_{i',j'}\}$ have density at most $\rho_i(\pi_i(\widehat{\tau}_i))$.
    Similarly, all chores in the set $supp(X_{i'} - Y) = \{j' \in [m+1]: x_{i',j'} > y_{j'}\}$ have density at least $\rho_i(\pi_i(\widehat{\tau_i}))$.

    Note that $X_{i'} = (X_{i'} - Y) + (Y \cap X_{i'})$ and $Y = (Y-X_{i'}) + (Y \cap X_{i'})$.
    Then, we have that
    \begin{align*}
        s(X_{i'} - Y) + s(Y \cap X_{i'}) 
        & = s(X_{i'}) \\
        & = B_{i'} \quad (\text{by definition of DD})\\ 
        & \geq s(Y) \quad \text{(by } (\ref{eqn:divisible_sayxb})) \\
        & = s(Y - X_{i'}) + s(Y \cap X_{i'}),
    \end{align*}
    giving us
    \begin{equation} \label{eqn:divisible_sizexby}
        s(X_{i'} - Y) \geq s(Y - X_{i'}).
    \end{equation}
    This means that agent $i$'s disutility for the fractional assignment $Y$ is at most his disutility for bundle of $i'$:
    \begin{align*}
        d_i(Y) & = d_i(Y \cap X_{i'}) + d_i(Y-X_{i'}) \\
        & \leq d_i(Y \cap X_{i'}) + \rho_i(\pi_i(\widehat{\tau}_i)) s(Y-X_{i'}) \\
        & \leq d_i(Y \cap X_{i'}) + \rho_i(\pi_i(\widehat{\tau}_i)) s(X_{i'} - Y) \quad (\text{by } (\ref{eqn:divisible_sizexby}))\\
        & \leq d_i(Y \cap X_{i'}) + d_i(X_{i'} - Y)
        = d_i(X_{i'}).
    \end{align*}
    Therefore, we get that in the DD allocation $\mathcal{X}$, every agent $i \in [n]$ is EF towards all other agents.

    Next, we prove the EF property by the houseekeeper towards each agent.

    Let $x_{n+1,c}$ denote the fraction of chore $c_j$ allocated to the housekeeper, i.e., $x_{n+1,j} = 1 - \sum_{i'=1}^n x_{i',j}$.
    Fix any agent $i \in [n]$ and any chore $c_{\widehat{j}} \in I_i(\widehat{\tau})$.
    Note that condition (iii) of the Definition \ref{def:dd} gives us $1 - \sum_{i'=1}^n x_{i',j} = 0$ for all chores $c_j \in I(\widehat{\tau})$.
    Thus, for all internal chores $c_{\widehat{j}} \in I_i(\widehat{\tau})$, 
    \begin{equation*}
        x_{i,\widehat{j}} \geq x_{n+1,\widehat{j}} = 0.
    \end{equation*}

    Fix any agent $i \in [n]$ and fractional assignment $Y \leq X_{n+1}$ with the property that
    \begin{equation} \label{eqn:divisible_xachousekeeper}
        s(Y) \leq B_i.
    \end{equation}
    To prove that $d_i(Y) \leq d_i(X_i)$, we consider the fractional assignments $(X_i-Y) \in [0,1]^{m+1}$ and $(Y-X_i) \in [0,1]^{m+1}$.
    Additionally, write $\widehat{\tau} = (\widehat{\tau}_1,\dots,\widehat{\tau}_n)$ to  denote the integer vector that certifies the DD property of $\mathcal{X}$.

    Note that DD implies that, for all internal chores $c_j \in I_i(\widehat{\tau})$, we have $x_{i,j} \geq x_{n+1,j} \geq y_j$.

    Using this inequality and the definition of $I_i(\widehat{\tau})$, we obtain that all the chores in the set $supp(Y-X_i) = \{j' \in [m+1]: y_{j'} > x_{i,j'}\}$ have density at most $\rho_i(\pi_i(\widehat{\tau_i}))$.
    The DD property also mandates that all chores in the set $supp(x_i)$ have density at least $\rho_i(\pi_i(\widehat{\tau}_i))$.
    Since $supp(X_i-Y) \subseteq supp(X_a)$, the density of every chore in $supp(X_a - Y)$ is at least $\rho_i(\pi_i(\widehat{\tau}_i))$.

    Note that $X_i = (X_i - Y) + (Y \cap X_i)$ and $Y = (Y - X_i) + (Y \cap X_i)$.
    Then, we get that
    \begin{align*}
        s(X_i-Y) + s(Y \cap X_i) 
        & = s(X_i) \\
        & = B_i \quad (\text{by definition of DD}) \\
        & \geq s(Y) \quad (\text{by } (\ref{eqn:divisible_xachousekeeper})) \\
        & = s(Y-X_i) + s(Y \cap X_i),
    \end{align*}
    giving us
    \begin{equation} \label{eqn:divisible_sizexay}
        s(X_i-Y) \geq s(Y-X_i).
    \end{equation}
    This means that agent $i$'s disutility for the fractional assignment $Y$ is at most his disutility for his own bundle:
    \begin{align*}
        d_i(Y) & = d_i(Y \cap X_i) + d_i(Y-X_i) \\
        & \leq d_i(Y \cap X_i) + \rho_i(\pi_i(\widehat{\tau}_i)) s(Y-X_i) \\
        & \leq d_i(Y \cap X_i) + \rho_i(\pi_i(\widehat{\tau}_i)) s(X_i-Y) \quad (\text{by } (\ref{eqn:divisible_sizexay}))\\
        & \leq d_i(Y \cap X_i) + d_i(X_i - Y) \\
        & = d_i(X_i)
    \end{align*}
    Thus, the lemma is proven.
    $\hfill \square$
\end{proof}
Note that the above lemma and Proposition \ref{prop:divisible_efhcremove} imply that from a DD allocation $\mathcal{X}$ (and by removing the $(m+1)$-th chore from consideration), one obtains an EF allocation for the underlying instance.

\subsection{EF Algorithm}
To capture the DD property and develop our algorithm, we first define a family linear programs, $LP_1(\cdot)$ that are parameterized by integer vector $\tau \in \mathbb{Z}^n_+$.

\begin{definition}
    Given any integer vector $\tau = (\tau_1,\dots,\tau_n)$ with $||\tau||_\infty \leq m+2$, 
    we define the following linear program, $LP_1(\tau)$, over decision variabes $\{ z_{i,j} \in [0,1] \}_{i,j}$:
    \begin{enumerate}[(i)]
        \item $z_{i,j} \leq z_{i',j}$ for all $i,i' \in [n]$ and $c_j \in I_i(\tau)$
        \item $\sum_{j:c_j \in I_i(\tau) \cup E_i(\tau)} z_{i,j} s(c_j) = B_i$ for all $i \in [n]$
        \item $\sum_{i=1}^n z_{i,j} = 1$ for all $c_j \in I(\tau)$
        \item $z_{i,h} = 0$ for all $i \in [n]$ and $c_h \neq I_i(\tau) \cup E_i(\tau)$
        \item $\sum_{i=1}^n z_{i,h} \leq 1$ for all $c_h \in C \cup \{c_{m+1}\} \setminus I(\tau)$
    \end{enumerate}
\end{definition}
Intuitively, $\tau = (\tau_1,\dots,\tau_n)$ captures how many items agent $i$ can get from the first $\tau_i$ densest item in the bundle.

The first three conditions are the requirements from the definition of density domination. The next two constraints ensure that the $z_{i,j}$'s induce an allocation.
The fourth set of constraints (iv) are redundant since they follow from the second set of constraints (ii), but we need them for the relaxation later.
Moreover, as stated in the following proposition, the feasibility of $LP_1(\cdot)$ implies the existence of DD allocations.
The proof is straightforward and hence omitted.
\begin{proposition} \label{prop:divisible_feasibleDD}
    If for an integer vector $\tau \in \mathbb{Z}^n_+$, the linear program $LP_1(\tau)$ is feasible, then a feasible solution $\{ z_{i,j}\}_{i,j}$ of $LP_1(\tau)$ corresponds to an allocation $z \in [0,1]^{n \times (m+1)}$ that satisfies the density domination property.
\end{proposition}
Note that it is unclear that $LP_1(\tau)$ is feasible for any integer vector $\tau \in \mathbb{Z}^n_+$.
However, if there exists a $\tau$ that induces a feasible $LP_1(\tau)$, then by Proposition \ref{prop:divisible_feasibleDD} and Lemma \ref{prop:divisible_efhc}, we will obtain the desired guarantee for the existence of EF.
We now show that such a $\tau$ exists and can be computed in polynomial time.

We formulate a new family of linear programs, $LP_2(\tau)$ (again, parameterized by integer vectors), by relaxing the second set of constraints of $LP_1(\tau)$, i.e., we relax the requirement that for every agent the budget constraint holds with an equality.

\begin{definition}
    Given any integer vector $\tau = (\tau_1,\dots,\tau_n) \in \mathbb{Z}^n_+$ (with $||\tau||_\infty \leq m+2$), we define the following linear program, $LP_2(\tau)$ over the decision variables $\{z_{i,j} \in [0,1]\}_{i,j}$:
    \begin{enumerate}[(i)]
        \item $z_{i,j} \leq z_{i',j}$ for all $i,i' \in [n]$ and $c_j \in I_i(\tau)$
        \item $\sum_{j:c_j \in I_i(\tau) \cup E_i(\tau)} x_{i,j} s(c_j) \leq B_i$ for all $i \in [n]$
        \item $\sum_{i=1}^n z_{i,j} = 1$ for all $c_j \in I(\tau)$
        \item $z_{i,h} = 0$ for all $i \in [n]$ and $c_h \in C \cup \{c_{m+1}\} \setminus (I_i(\tau) \cup E_i(\tau))$
        \item $\sum_{i=1}^n z_{i,h} \leq 1$ for all $c_h \in C \cup \{c_{m+1}\} \setminus I(\tau)$
    \end{enumerate}
\end{definition}
Now, we prove an important lemma.
For each $k \in [n]$, we denote $e_k$ as the $k$-th standard basis vector in $\mathbb{R}^m$.

\begin{lemma} \label{lem:divisible_LP2feasibleek}
    Let $\tau \in \mathbb{Z}^n_+$ be an integer vector with $||z||_\infty \leq m+1$.
    If the linear program $LP_2(\tau)$ is feasible and $LP_1(\tau)$ is infeasible, then there exists an agent $k \in [n]$ such that $LP_2(\tau + e_k)$ is feasible.
 \end{lemma}
 \begin{proof}
     For vector $\tau \in \mathbb{Z}^n_+$, given that the program $LP_2(\tau)$ is feasible, we consider among its feasible solutions one that maximizes $\sum_{i=1}^n z_{i, \pi_i(\tau_i)}$. Write $z^* = \{z^*_{i,j} \}_{i,j}$ as a feasible solution that maximizes the fractional assignment across the edge chores $\{\pi_i(\tau_i)\}_i$ among all feasible solutions of $LP_2(\tau)$.

     Consider an agent $i'$ for whom the budget constraint holds with a strict inequality, i.e., 
     \begin{equation}\label{inequality:primei}
         \sum_{j: c_j \in I_{i'}(\tau) \cup E_{i'}(\tau)} z^*_{i',j} s(c_j) < B_{i'}.
     \end{equation}
     Note that such an agent $i'$ necessarily exists, otherwise $LP_1(\tau)$ would be feasible.

     Let $c_{\widehat{j}} = \pi_{i'}(\tau_{i'})$ be the edge chore of agent $i'$, i.e., $E_{i'}(\tau) = \{c_{\widehat{j}}\}$.
     First we claim that
     \begin{equation} \label{eqn:divisible_=1}
         \sum_{i=1}^n z^*_{i,\widehat{j}} = 1.
     \end{equation}
     Suppose for a contradiction that $\sum_{i=1}^n z^*_{i,\widehat{j}} < 1$.
     This strict inequality and the feasibility of $z^*$ implies that $c_{\widehat{j}}$ is not an internal chore for any agent, i.e., $c_{\widehat{j}} \notin I(\tau)$ (by condition (iii) of the LP).
     Hence, $c_{\widehat{j}}$ does not participate in (i) (the first set of constraints) in $LP_2(\tau)$.

     In such a case, we can increment $z^*_{i',\widehat{j}}$ by a sufficiently small $\epsilon > 0$, while maintaining feasibility and, in particular, satisfying the budget constraint of agent $i'$.
     Such an update, however, increases the objective function $\sum_{i=1}^n z_{i,\pi_i(\tau_i)}$ and hence contradicts the fact that $z^*$ maximizes this objective function.
     Therefore, (\ref{eqn:divisible_=1}) holds for chore $c_{\widehat{j}}$.

     Now, let $N$ denote the set of agents who have received a nonzero fraction of chore $c_{\widehat{j}}$ under $z^*$, i.e., $N = \{i \in [n] : z^*_{i,\widehat{j}}  > 0 \}$.
     Equation (\ref{eqn:divisible_=1}) ensures that $N \neq \emptyset$.

     Let $M \subseteq N$ denote the set of agents in $N$ who have received the minimum fraction of $c_{\widehat{j}}$, i.e., $M = \{i \in N: z^*_{i,\widehat{j}}=w^*\}$ where $w^*= \min_{i \in N: z^*_{i,\widehat{j}}>0} z^*_{i,\widehat{j}}$. 
     Also let $F_E$ and $F_I$ be the set of agents for whom $c_{\widehat{j}}$ is an edge chore and internal chore respectively:
     \begin{equation*}
         F_E := \{ i\in [n] : c_{\widehat{j}} \in E_i(\tau)\}, F_I := \{i\in [n]:c_{\widehat{j}} \in I_i(\tau)\}.
     \end{equation*}
     Indeed, the above-identified agent $i'$ is necessarily contained in $F_E$.
     Additionally, via (iv) in $LP_2(\tau)$, for every agent $i \in [n]$, either $c_{\widehat{j}} \in E_i(\tau)$ or $c_{\widehat{j}} \in I_i(\tau)$.
     Hence, $N \subseteq F_I \cup F_E$.

     We now want to show $F_E$ and $M$ intersects.

     \begin{claim}
         $F_E \cap M \neq \emptyset$.
     \end{claim}
     \begin{proof}
         Suppose for a contradiction $F_E \cap M = \emptyset$.
         Then, since $F_I \subseteq M$ by (i) in $LP_2(\tau)$ and $M \subseteq N \subseteq F_I \cup F_E$, we obtain $M = F_I$. This equality then shows that $F_I \neq \emptyset$.
        Recall that for agent $i'$, under $z^*$, the budget constraint was not tight. 
        By \eqref{inequality:primei} and by the fact that $w^*>0$, we can select a sufficiently small $\epsilon$ with 
        $$
        0< \epsilon \leq \min \{B_{i'}-\sum_{j: c_j \in I_{i'}(\tau) \cup E_{i'}(\tau)} z^*_{i',j} s(c_j), w^* \}.
        $$
        Since each agent in $F_I$ receives a positive fraction of $c_{\widehat{j}}$, we can update $z^*$ to obtain another feasible $z'$ as follows:
         \begin{itemize}
             \item set $z'_{p,\widehat{j}} = z^*_{p,\widehat{j}} - \frac{\epsilon}{|F_I|}$, for all agents $p \in F_I$; and $z'_{i',\widehat{j}} = z^*_{i',\widehat{j}} + \epsilon$
             \item For all other agent-chore pairs, the fractional assignment in $z'$ is the same as in $z^*$.
         \end{itemize}
         We now show that the solution $z'$ is feasible for $LP_2(\tau)$. 

         \begin{itemize}
         \item The (i) constraints in $LP_2(\tau)$ are maintained since $M=F_I$ and we have decreased the fractions of $c_{\widehat{j}}$ allocated the agents in $F_I$. 
         \item Clearly, increasing $z^*_{i',\widehat{j}}$ by $\epsilon$ (as in $z'$) maintains feasibility with respect to the (ii) constraint with respect to $i'$ and every other agent.

         \item Chore $c_{\widehat{j}}$ continues to be fully assigned among the agents, since we have cumulatively reduced the fractional assignments of $c_{\widehat{j}}$ among the agents in $F_I$ by $\epsilon$, and increased the assignment to $i'$ by $\epsilon$.
         Therefore, the (iii) constraints continue to hold for $z'$.

         \item Furthermore, the (iv) and (v) constraints in $LP_2(\tau)$ also hold for $z'$.
         \end{itemize}
         Therefore, $z'$ satisfies all the constraints in $LP_2(\tau)$.
         However, recall that $\pi_{i'}(\tau_{i'}) = c_{\widehat{j}}$, and hence $z'$ strictly increases the fractional assignment across the edge chores $\{\pi_i(\tau_i)\}_i$ than 
         $z^*$. This contradicts the optimality of $z^*$. Hence, by way of contradiction, we obtain the stated claim, $F_E \cap M \neq \emptyset$.
         $\hfill \square$
    \end{proof}
    We will now complete the proof using this claim.

    In particular, we will show that, for any agent $k \in F_E \cap M$, the program $LP_2(\tau + e_k)$ is feasible; in fact, $z^*$ itself is a feasible solution for $LP_2(\tau + e_k)$.

    Fix any agent $k \in F_E \cap M$.
    Since $k \in F_E$, we have $I_k(\tau + e_k) = I_k(\tau) \cup \{c_{\widehat{j}} \}$.
    Additionally, $k \in M = \argmin_i z^*_{i,\widehat{j}}$ and hence, solution $z^*$ satisfies the (i) constraints in $LP_2(\tau + e_k)$, for agent $k$ and other relevant agents as well.
    The (ii) constraints in $LP_2(\tau + e_k)$ are the same as in $LP_2(\tau)$, hence $z^*$ continues to be feasible with respect to these budget constraints.

    In addition, equation (\ref{eqn:divisible_=1}) enforces the (iii) constraints for chores $c_{\widehat{j}}$.
    The (iii) constraints hold for all other chores in $I(\tau + e_k)$---this follows from the fact that $I(\tau + e_k) = I(\tau) \cup \{ c_{\widehat{j}}\}$ and $z^*$ is a feasible solution with respect to $I(\tau)$.

    Finally, using the fact that $I_i(\tau + e_k) \cup E_i(\tau + e_k) \supseteq I_i(\tau) \cup E_i(\tau)$ for all agents $i$, we obtain that the (iv) and (v) constraints are satisfied by $z^*$ in $LP_2(\tau + e_k)$ as well.
    Overall, we obtain that $LP_2(\tau + e_k)$ is feasible and the lemma stands proved.
    $\hfill \square$
 \end{proof}
The following lemma shows that $LP_2(\cdot)$ cannot be incessantly feasible.
\begin{lemma} \label{lem:divisible_LP2infeasiblem+2}
    For any integer vector $\tau \in \mathbb{Z}^n_+$ with $||\tau||_\infty = m+2$, the program $LP_2(\tau)$ is infeasible.
\end{lemma}
\begin{proof}
    Given that, for the given vector $\mathbb{Z}^n_+$, one of the components is equal to $m+2$, we have that the chore $c_{m+1} \in I(\tau)$; see definition of internal and edge chores.

    Now, for $LP_2(\tau)$ to be feasible, the chore $c_{m+1}$ must be fully assigned among the agents; see (iii) constraints in $LP_2(\tau)$.

    However, since $s(c_{m+1}) = 2n \cdot \max_{i'} B_{i'}$ for all agents $i \in [n]$, such an assignment is not possible while maintaining the budget constraints (ii) of the agents.
    The lemma stands proved.
    $\hfill \square$
\end{proof}

\begin{algorithm}
\caption{Computing EF allocations for divisible chores}
\label{alg:divisible_ef}
\SetInd{0.8em}{0.3em}
Initialize an $n$-dimensional integer vector 
$\tau \leftarrow (1,1,\dots,1)$\;
\While{$LP_1(\tau)$ is infeasible}{
Find $k \in [n]$ such that $LP_2(\tau + e_k)$ is feasible (such an agent always exists by Lemma \ref{lem:divisible_LP2feasibleek})\;
Update $\tau \leftarrow \tau + e_k$\;
}
\Return allocation $\mathcal{X} = (X_1,\dots,X_n)$ corresponding to a feasible solution of $LP_1(\tau)$\;
\end{algorithm}

We now proceed to prove that Algorithm \ref{alg:divisible_ef} returns an EF allocation.

Our approach is to prove that constraints (i) to (v) of the linear program $LP_2(\tau)$ is satisfied at every iteration.
For the initial vector $\tau = (1,\dots,1)$, since for all agents $i \in [n]$, $I_i(\tau) = \emptyset$.
By selecting the all-zero solution, all constraints are trivially satisfied.

Thus, at the beginning of each iteration of the \textbf{while} loop, $LP_2(\tau)$ is feasible for the vector $\tau$.
Now, if for the current $\tau$, $LP_1(\tau)$ is feasible, then we return an EF allocation (the existence is guaranteed by Lemma \ref{prop:divisible_efhc} and Proposition \ref{prop:divisible_feasibleDD}).
Otherwise, if $LP_1(\tau)$ is infeasible, then the loop executes and we update $\tau$ to $\tau + e_k$.
By Lemma \ref{lem:divisible_LP2feasibleek}, we are guaranteed to be able to find a $k \in [n]$ such that $LP_2(\tau + e_k)$ is feasible.

Finally, note that the \textbf{while} loop cannot iterate indefinitely---after at most $n\cdot(m+1)$ iterations, $\tau$ will increment to satisfy $||\tau||_\infty = m+2$.
However, by Lemma \ref{lem:divisible_LP2infeasiblem+2}, we know that for such a $\tau$, the program $LP_2(\tau)$ is infeasible.
Thus, the \textbf{while} loop  will terminate in $\mathcal{O}(nm)$ iterations and the algorithm returns an EF allocation.

Since each iteration of the \textbf{while} loop involves solving a polynomially-large number of linear programs, the runtime of the algorithm is polynomially-bounded, and we obtain our result.

\end{document}